\title{Robust testing in generalized linear models by sign-flipping score contributions\footnote{To appear in \emph{Journal of the Royal Statistical Society: Series B (Statistical Methodology)}.}}

\documentclass[11pt]{article}

\usepackage{amsmath}
\usepackage{amsfonts}
\usepackage{bm}
\usepackage{bbm}
\usepackage{amsthm}
\usepackage{comment}
\usepackage[nottoc]{tocbibind}
\usepackage{graphicx}
\usepackage{subcaption}
\usepackage{changepage}
\usepackage{natbib}
\usepackage{abstract}
\usepackage[colorlinks=true,citecolor=blue,runcolor=blue,linkcolor=blue, pdfborder={0 0 0}]{hyperref}

\usepackage{tikz}
\usepackage{pgfplots}
\usepackage{appendix}

\usepackage[margin=3.5cm]{geometry}

\definecolor{cobalt}{rgb}{0.0, 0.28, 0.67}
\definecolor{darkblue}{rgb}{0.0, 0.0, 0.55}
\definecolor{dgreen}{rgb}{0.0, 0.4, 0.0}

\theoremstyle{plain}
\newtheorem{theorem}{Theorem}
\newtheorem{corollary}[theorem]{Corollary}
\newtheorem{lemma}[theorem]{Lemma}
\newtheorem{proposition}[theorem]{Proposition}

\theoremstyle{definition}

\newtheorem{assumption}{Assumption}

\newcommand{\Ima}{\hat{\bm{ \mathcal{I} }}}
\newcommand{\Imag}{\bm{ \mathcal{I} }}
\newcommand{\Imat}{\tilde{\bm{ \mathcal{I} }}}
\newcommand{\fatS}{\bm{S}}
\newcommand{\fnu}{\bm{\nu}}
\newcommand{\fgamma}{\bm{\gamma}}
\newcommand{\fbeta}{\bm{\beta}}

\newcommand{\black}{\color{\black}}
\newcommand{\lp}{\eta} 

  \author{Jesse Hemerik\footnote{Department of Biostatistics, Oslo Centre for Biostatistics and Epidemiology, University of Oslo} \footnote{Address for correspondence: Jesse Hemerik, Oslo Centre for Biostatistics and Epidemiology, P.O.Box 1122 Blindern,
0317 Oslo, Norway. e-mail: jesse.hemerik@medisin.uio.no}, 
Jelle Goeman\footnote{Biomedical Data Sciences, Leiden University Medical Center, the Netherlands}\phantom{.}
and 
    Livio Finos\footnote{Department of Developmental Psychology and Socialization, University of Padua, Italy}
       }

\begin{document}
\maketitle

\begin{abstract}
Generalized linear models are often misspecified due to overdispersion, heteroscedasticity and ignored nuisance variables. 
Existing quasi-likelihood methods for testing in misspecified models often do not provide satisfactory type-I error rate   control.
We provide a novel semi-parametric test, based on sign-flipping individual score contributions. The tested parameter is allowed to be multi-dimensional and even high-dimensional. 
Our test  is  often robust against  the mentioned forms of misspecification and provides better type-I error control than its competitors. 
When nuisance parameters are estimated, our basic test becomes conservative. We show how to take nuisance estimation into account to  obtain an asymptotically exact test. Our proposed test is asymptotically equivalent to its parametric counterpart. 

\end{abstract}

\section{Introduction}

We consider the problem of  testing hypotheses about parameters in potentially misspecified generalized linear models (GLMs). The types of misspecification that we consider include overdispersion and heteroscedasticity.
When the model is misspecified, the traditional parametric tests tend to lose their properties,  for example because they estimate  the Fisher information under incorrect assumptions.  
By a parametric test we mean a test which fully relies on an assumed parametric model \citep{pesarin2015some} to compute the null distribution of the test statistic.

When a parametric model to be tested is potentially misspecified, the most obvious approach is to extend the model with more parameters, e.g. to add an overdispersion parameter. However, such approaches still require assumptions, for example that the overdispersion is constant. Hence a fully parametric approach is not always the best option.

Another well-known approach to testing in possibly misspecified GLMs is to use a Wald-type test, where a sandwich estimate of the variance of the coefficient estimate is used. The sandwich estimate corrects for the potentially misspecified variance. As long as the linear predictor and link are correct, such a test is asymptotically exact under mild assumptions.
 We call a test asymptotically exact if its  rejection probability is asymptotically known under the null hypothesis. 
 For small samples, however, sandwich estimates often perform poorly and the test can be very liberal \citep{boos1992on,freedman2006so,maas2004robustness,kauermann2000sandwich}.

Recent decades have seen an increase in the use of permutation approaches for various testing problems \citep{tusher2001significance,pesarin2001multivariate,chung2013exact,pauly2015asymptotic,winkler2016faster,hemerik2018exact, ganong2018permutation}. 
These methods are useful since they require few parametric assumptions.
Especially when multiple hypotheses are tested, permutation methods are often powerful since they can take into account the dependence structure in the data \citep{westfall1993resampling, hemerik2018false,hemerik2019permutation}. In the past, permutation methods have already been used to test in linear models \citep[][and references therein]{winkler2014permutation}. Rather than permutations, sometimes other transformations are used,  such as rotations \citep{solari2014rotation} and sign-flipping of residuals \citep{winkler2014permutation}.
The existing permutation tests for GLMs, however, are limited to models with identity link function.

Like some existing methods for testing in linear models, this paper presents a sign-flipping approach. Our approach is new however, since rather than flipping residuals, we flip individual score contributions (note that the score, the derivative of the log-likelihood, is a sum of $n$ individual score contributions). Moreover, we allow testing in a wide range of models, not only regression models with identity link.
Under mild assumptions, the only requirement for the test to be asymptotically exact, is that the individual score contributions have mean 0.
Consequently, if the link function is correct, our method is  often  robust against several types of model specification, such as arbitrary overdispersion, heteroscedasticity and, in some cases, ignored nuisance parameters.

The main reason for this robustness is that we do not require to estimate the variance of the score, the Fisher information. Rather, we perform a permutation-type test based on the score contributions, where rather than permutation, we use sign-flipping. Intuitively, the advantage of this  approach over explicitly estimating the variance, is the following: if the score contributions  are independent and perfectly symmetric around zero under the null, then our test is exact for small $n$, even if  the score contributions have misspecified variances  and shapes  \citep{pesarin2010permutation}. A parametric test, on the other hand, is then usually not exact.

In case nuisance  parameters  are estimated, the individual score contributions become dependent and our basic sign-flipping test is no longer asymptotically exact. 
To deal with this problem, we consider the \emph{effective score}, which is less dependent on the nuisance estimate than the basic score \citep{hall1990large,marohn2002comment}.
In this case we need slightly more assumptions: the variance misspecification is not always allowed to depend on the covariates.
The resulting test is  asymptotically exact. 

The methods in this paper have been implemented in the R package \emph{flipscores}, available on CRAN.

In Section \ref{knownn} we consider the scenario that no nuisance effects need to be estimated. In Section \ref{unknownn} we show how the estimation of nuisance effects can be taken into account.
 Section \ref{secmv} provides tests of hypotheses about parameters of more than one dimension. 
 Section \ref{simulations} contains simulations and Section \ref{data} an analysis of real data.

\section{Models with known nuisance parameters} \label{knownn}
Consider random variables $\nu_1,...,\nu_n$, which satisfy Assumption \ref{as:nu} below. These will often be individual score contributions (see Section \ref{unknownn},  \citeauthor{rao1948large}, \citeyear{rao1948large}, or \citeauthor{hall1990large},  \citeyear{hall1990large},    p. 86), but the results in Section \ref{secbasict}  hold for any random variables satisfying this assumption.

\begin{assumption} \label{as:nu}
  The random variables $\nu_i,$ $i\in \mathbb{N}$,  are independent of each other, have finite variances and satisfy the following.
For every $\epsilon>0$, $$\lim_{n\rightarrow\infty}  \frac{1}{n} \sum_{i=1}^{n}\mathbb{E}\Big(\nu_i^2\mathbbm{1}_{\{ \nu_i /\sqrt{n}>\epsilon  \}}\Big)=0.$$
Further,  as $n\rightarrow\infty$, $s_n^2:= \frac{1}{n}\sum_{i=1}^{n}var(\nu_i) \rightarrow s^2$ for some constant  $s^2> 0$. 

\end{assumption}

Throughout Section \ref{knownn}, we consider any null hypothesis $H_0$ which implies that $\mathbb{E}\nu_i=0$ for all $1\leq i \leq n$. 
If $\nu_1,...,\nu_n$ are score contributions and $H_0$ is a point hypothesis, then under mild  assumptions, $\mathbb{E}\nu_i=0$ is satisfied under $H_0$.

A key assumption throughout Section \ref{knownn} is that the $\nu_i$, $i\in \mathbb{N}$, are independent. As soon as nuisance parameters need to be estimated, however,  score contributions become dependent. Section \ref{unknownn} is devoted to dealing with estimated nuisance.

\subsection{Basic sign-flipping test} \label{secbasict}
Let  $\alpha\in[0,1)$.
For any $a\in \mathbb{R}$, let $\lceil a\rceil$ be the smallest integer which is larger than or equal to $a$ and let $\lfloor a\rfloor$ be the largest integer which is at most $a$. 
Given  values $T_1^n,...,T_w^n\in \mathbb{R}$, we let
 $T_{(1)}^n\leq...\leq T_{(w)}^n$ be the sorted values and write $T_{[1-\alpha]}^n=T_{(\lceil (1-\alpha)w \rceil)}^n$.

Throughout this paper, $w\in \{2,3,...\}$ denotes the number of random sign-flipping  transformations  to be used. Define $g_1=(1,...,1)\in \mathbb{R}^n$ and for every $2\leq j \leq w$ let $g_j=(g_{j1},...,g_{jn})$ be independent and uniformly distributed on $\{-1,1\}^n$. 
Throughout the rest of Section \ref{knownn}, for every $1\leq j \leq w$,  we let  
$$T_j^n= n^{-1/2} \sum_{i=1}^n g_{ji} \nu_i. $$
We now state that the basic sign-flipping test is asymptotically exact  for the point null hypothesis $H_0$ that implies $\mathbb{E}\nu_i=0$, $1\leq i \leq n$. All proofs are in the appendices.

\begin{theorem} \label{mainST}
Suppose that Assumption \ref{as:nu} holds.
Consider the test that rejects $H_0$ if and only if $T_1^n>T_{[1-\alpha]}^n$. 
Then, as $n\rightarrow \infty$, the  rejection probability  of this test converges to $\lfloor \alpha w \rfloor/w\leq \alpha$ under $H_0$. Moreover, the statistics $T_1^n,...,T_w^n$ are asymptotically  normal and independent with mean $0$ and common variance  $\lim_{n\rightarrow\infty}s^2_n$ under $H_0$.
\end{theorem}

 We now provide an extension of Theorem \ref{mainST} to interval hypotheses. The proof is a straightforward adaptation of the proof of Theorem  \ref{mainST}.

 \begin{corollary}[Interval hypotheses]   \label{mainSTinterval}
Suppose Assumption \ref{as:nu} holds.
 Consider a null hypothesis $H'$ which implies that $\mathbb{E}\nu_i \leq 0$ (respectively $\mathbb{E}\nu_i \geq 0$) for all $1\leq i \leq n$. 
Then for every $\epsilon >0$, there exists an $N\in\mathbb{N}$ such that under $H'$, for every $n>N$,
 $\mathbb{P}\big\{T_1^n>T_{(\lceil(1-\alpha)w \rceil)}^n\big\}$ (respectively $\mathbb{P}\big\{T_1^n <T_{(\lfloor \alpha w+1 \rfloor )}^n\big\}$)  is at most $\lfloor \alpha w \rfloor/w +\epsilon$. 
\end{corollary}

The following corollary extends Theorem \ref{mainST}  to two-sided tests. The proof is analogous to that of Theorem \ref{mainST}.

\begin{corollary}[Two-sided test]   \label{mainSTts}
Suppose Assumption \ref{as:nu} holds.
Consider $\alpha_1,\alpha_2 \in \{0/w, 1/w,...,(w-1)/w\}$. 
Under $H_0$, as $n\rightarrow \infty$,
$$\mathbb{P}\Big[\big \{T_1^n<T_{(\alpha_1 w+1)}^n\big \} \cup \big\{T_1^n>T_{((1-\alpha_2)w)}^n\big\}\Big] \rightarrow \alpha_1+\alpha_2.$$
\end{corollary}

Note that our test does not rely on an approximate symmetry assumption \citep[as e.g.][]{canay2017randomization}. Indeed, even if the scores are very skewed, asymptotically the test of Theorem \ref{mainST} is exact.
However,  if the $\nu_i$ are symmetric, then even for small $n$
the  size   is always at most $\alpha$, as noted in the following proposition.   A special case of this result is already discussed in \citet[][\S 21]{fisher1935}, where every element of $\{1,-1\}^n$ is used once.

\begin{proposition} \label{exacttest}
Suppose that $\nu_1,...,\nu_n$ are independent and continuous and that under $H_0$, for each $1\leq i \leq n$, $\nu_i \,{\buildrel d \over =}\,   -\nu_i$. Then the  size  of the test of Theorem \ref{mainST} is at most $\alpha$ for any $n\in \mathbb{N}$.
Moreover, if  $g_2,...,g_w$ are uniformly drawn from $\{1,-1\}^n\setminus (1,...,1)$ without replacement (so that only $g_1$ takes the  value $(1,...,1)$), then the  rejection probability under $H_0$  is exactly $\lfloor \alpha w \rfloor /w$. (Note that $w$ cannot exceed $2^n$ then.) 
\end{proposition}

If the $g_j$ are drawn with replacement or the $\nu_i$ are discrete, then  under $H_0$ the rejection probability  of the test of Proposition \ref{exacttest} is (slightly) smaller than $\lfloor \alpha w \rfloor /w$ for finite $n$, due to the possibility of ties among the test statistics $T_j^n$, $1\leq j \leq w$.  Otherwise the  rejection probability under $H_0$   is $\lfloor \alpha w \rfloor /w$.

Note that when the rejection probability under $H_0$ is $\lfloor \alpha w \rfloor/w$, it can be advantageous to take $w$ such that $\alpha$ is a multiple of $1/w$, to exhaust  the nominal level.

In Theorem \ref{mainST}, we  did not assume continuity of the observations $\nu_i$. 
There,  under the mild Assumption \ref{as:nu},  for $n\rightarrow \infty$, $\mathbb{P} (T_j^n=T_k^n )\rightarrow 0$ for any $1\leq j<k\leq w$, regardless of the distribution of the $\nu_i$. This allows using Theorem \ref{mainST} for discrete GLMs.

\subsection{Robustness} \label{secrob}
As a main example we consider   the exponential  family, i.e.,
suppose independent variables $Y_1,...,Y_n$ have densities of the form
$$
f(y_i;\eta_i) = \exp\bigg \{ \frac{y_i \lp_i-b(\lp_i)}{a_i} +c(y_i)   \bigg\},
$$
 where $\lp_i=x_i\beta +\bm{z}_i{\fgamma}$, $x_i, \beta \in \mathbb{R}$,  $\bm{z}_i, {\fgamma}\in \mathbb{R}^{m}$ for some $m\in\mathbb{N}$. Here $\beta$ is the coefficient of interest and presently we assume the other covariates $\fgamma$ to be known.
The canonical link function $g$ satisfies $\lp_i=g(\mu_i)$, where $g^{-1}(\lp_i)=\mu_i=\mathbb{E}(y_i)=b'(\lp_i)$ and $a_i=\text{var}(y_i)/b''(\lp_i)$ \citep{agresti2015foundations}.
For $H_0: \beta =\beta_0$, the score $\sum_{i=1}^{n} \nu_i =  \sum_{i=1}^n\frac{\partial}{ \partial \beta} \log\{f(y_i;\eta_i)\} \vert_{\beta=\beta_0}$   is
$$
\sum_{i=1}^{n} \frac{x_i\big(y_i-b'(\lp_i)\big)}{a_i}  \Big\vert_{\beta=\beta_0}  =\sum_{i=1}^{n} \frac{x_i\big(y_i-\mathbb{E}(y_i)\big)}{a_i}  \Big\vert_{\beta=\beta_0}. 
$$

For example, the Poisson model  has link function $g=\log$,  $b(\lp_i)=\exp(\lp_i)$, $a_i=1$ and $c(y_i)=-\log(y_i!).$
Hence $\mathbb{E}(y_i)=b'(\lp_i)=\exp(\lp_i)$. Thus the score function is
$$  \sum_{i=1}^{n} x_i\big (y_i -\mu_i  \big )\vert_{\beta=\beta_0} = \sum_{i=1}^{n} x_i\big \{y_i -\exp( x_i\beta_0 +\bm{z}_i{\fgamma})   \big \}. $$
For the normal distribution, $a_i=\sigma^2$,  so that the score is
$$
 \sum_{i=1}^{n} \frac{x_i(y_i-\lp_i)}{\sigma^2}\Big \vert_{\beta=\beta_0}.
$$

Apart from some mild assumptions, the main assumption made in Theorem \ref{mainST} is that $\mathbb{E}(\nu_i)=0$, $i=1,...,n$. This is satisfied as soon as $\mu_i\vert_{\beta=\beta_0}$ is the true expected value of $Y_i$. Then the test is asymptotically exact even if the $a_i$ are misspecified, i.e., if the variance or distributional shape of $Y_i$ is misspecified.
The $a_i$ are even allowed to be misspecified by a factor which depends on the covariates, as long as Assumption \ref{as:nu} holds.

As a concrete example, consider the normal model with identity link function, which assumes that $var(Y_1)=...=var(Y_n)$. If the real distribution is heteroscedastic,  then the test will still be exact for finite $n$, since the $\nu_i$ are symmetric.
The parametric test, however, loses its properties, for example because  the estimated variance does not have the assumed chi-squared distribution.
In Section \ref{simulations} it is illustrated that our approach can be much more robust against heteroscedasticity than a parametric test.

Another example is the situation where the model is Poisson, i.e., $var(Y_i)=\mu_i$ is assumed, but  in reality $var(Y_i)>\mu_i$, a form of overdispersion which occurs very often in practice. Then the  parametric score test underestimates the Fisher information and is anti-conservative. To take the overdispersion into account it could be explicitly estimated. However, if the overdispersion factor is not constant, but  depends on the covariates, then again the parametric test loses its properties. Theorem \ref{mainST}, however, often still applies, so that an asymptotically exact test is obtained. 

Further, note that if $\mathbb{E}(Y_i)$ depends on a nuisance variable $Z^l_{i}$ which is latent and ignored, where $Z^l_{i}$ is independent of  $X_i$, then the test may still be valid. The reason is that marginal over $Z^l_{i}$, $\mathbb{E}(Y_i)$ may still be computed correctly (see, for example, Section \ref{secignn}). Such latent nuisance variables will increase the variance of $Y_i$, however, which can pose a problem for the classical parametric score test, which needs to compute the Fisher information. When the latent variable is not independent of $X_i$, this usually does pose a problem for our test (even as $n\rightarrow \infty$), since   $\mathbb{E}(Y_i-\mu_i)$ becomes dependent on $X_i$ under $H_0$.

\section{Taking into account nuisance estimation} \label{unknownn}

Consider  independent and identically distributed pairs $(\bm{X}_i,Y_i)$, $i=1,...,n$, where $\bm{X}_i$ is some covariate vector and  $Y_i\in \mathbb{R}$ has distribution $\mathbb{P}_{\beta,\fgamma_0,\bm{X}_i}$, which depends on parameter of interest  $\beta\in \mathbb{R}$ and unknown nuisance parameter $\fgamma_0$, which lies in a set  $\mathbb{G} \subseteq  \mathbb{R}^{k-1}$,  where $k$ is the total number of modeled parameters.  
We will discuss the issues arising from estimating  $\fgamma_0$ and propose a solution, which allows us to obtain an asymptotically exact test  based on score flipping. 
Note that in this paper, the model defined above is the model considered by the user. It is the model used to compute the scores. We consider this model to be correct, unless explicitly stated otherwise, for example in Section \ref{robeff}. The parameter  $\fgamma_0$ is part of the  model considered by the user, so it is always modeled and estimated. For example, $\fgamma_0$  never represents ignored overdispersion.

We consider the  null hypothesis $H_0: \beta=\beta_0 \in \mathbb{R}$.  Generalizations to interval hypotheses and two-sided tests can be obtained as in Corollaries \ref{mainSTinterval}  and \ref{mainSTts}.  The case that the parameter of interest is multi-dimensional is considered in Section \ref{secmv}.

 Suppose that for all $\fgamma\in\mathbb{G}$, $\mathbb{P}_{\beta,\fgamma,\bm{X}_i}$ has a density $f_{\beta,\fgamma,\bm{X}_i}$ around $\beta_0$ with respect to some dominating measure.  For $1\leq i \leq n$ write $$\nu_{\fgamma,i}= \frac{\partial }{\partial \beta } \log f_{\beta,{\fgamma},\bm{X}_i}(Y_i)\vert_{\beta=\beta_0},$$ where we assume the derivative exists. 
The value $\nu_i$ is the score for the $i$-th observation. 
Under $H_0$, $\mathbb{E} (\nu_{\fgamma_0,i})=0$, $i=1,...,n$.
The score for all $n$ observations simultaneously is $n^{1/2}S_{{\fgamma}}$, where
$$S_{{\fgamma}}=n^{-1/2}\sum_{i=1}^n \nu_{\fgamma,i}.$$

Assume that $\hat{\fgamma}$ is a $\sqrt{n}$-consistent estimate of $\fgamma_0$,  taking values in $\mathbb{G}$. 
For every $1\leq i \leq n$, let  $$ \fnu^{(k-1)}_{\hat{\fgamma},i}=    \frac{\partial }{\partial \fgamma} \log f_{\beta_0,{\fgamma},\bm{X}_i}({Y_i}) \Big\vert_{\fgamma=\hat{\fgamma}}  \in \mathbb{R}^{k-1} $$
denote the $(k-1)$-vector of score contributions for the nuisance parameters,  which is assumed to exist. 
Let $$\fatS^{(k-1)}_{\hat{\fgamma}}=n^{-1/2} \sum_{i=1}^{n}   \fnu^{(k-1)}_{\hat{\fgamma},i} \in \mathbb{R}^{k-1}$$ be the vector of nuisance scores.

For $1\leq j \leq w$, let the superscript $j$ denote that $g_j$ has been applied:
\begin{align*}
S^{j}_{\hat{\fgamma}}&=n^{-1/2} \sum_{i=1}^{n}   g_{ji}\nu_{\hat{\fgamma},i},\\
\fatS^{(k-1),j}_{\hat{\fgamma}}&=n^{-1/2} \sum_{i=1}^{n}   g_{ji}\fnu^{(k-1)}_{\hat{\fgamma},i}.
\end{align*}

\subsection{Asymptotically exact test} \label{seceffsctest}

When the nuisance parameter $\fgamma_0$ is unknown, it needs to be estimated, which is typically done by maximizing the likelihood of the data under the null hypothesis.
The  distribution of $S_{\hat{\fgamma}}$ can be substantially different from that of $S_{\fgamma_0}$, the score based on the true nuisance parameters. 
Indeed, under the null hypothesis, the asymptotic variance of $S_{\hat{\fgamma}}$ is not the Fisher information,  but the \emph{effective Fisher information} \citep [][section 9.3]{rippon2010generalised, rayner1997asymptotically, hall1990large,marohn2002comment, cox1979theoretical}, which is also the asymptotic variance of the \emph{effective score}, defined below.
The effective information is  smaller than the information, given that the score for the parameter of interest and the nuisance score are correlated. Intuitively, the reason is that the nuisance variable will be used to explain part of the apparent effect of the variable of interest, also asymptotically.

The estimation of $\fgamma_0$ makes the summands $\nu_{\hat{\fgamma},1},...,\nu_{\hat{\fgamma},n}$ underlying $S_{\hat{\fgamma}}$ correlated, in such a way that $var(S_{\hat{\fgamma}})<  var(S_{\fgamma_0})$ (if the score is correlated with the nuisance score). 
Note however that after random flipping, the summands are not correlated anymore. This means that the variance of $S_{\hat{\fgamma}}$ is asymptotically smaller than the variance of $S^j_{\hat{\fgamma}}$, $2\leq j \leq w$ (see the proof of Theorem \ref{signflipeff}). 
Hence, using $\nu_{\hat{\fgamma},1},...,\nu_{\hat{\fgamma},n}$ in the test of Theorem \ref{mainST} can lead to a conservative test, even as $n\rightarrow\infty$.

To make the test asymptotically exact again, we would like to adapt the individual scores such that they are less dependent on the random variation of $\hat{\fgamma}$.
We do this by considering the so-called \emph{effective score},  which ``is `less dependent' on the nuisance parameter than the usual score statistic" \citep[][p. 344]{marohn2002comment}.

The effective score $S^*_{\hat{\fgamma}} $ and the underlying summands $\nu^*_{\hat{\fgamma},i}$, $i=1,...,n$ (which we assume have nonzero variance for $\hat{\gamma}=\gamma_0$)  are defined as

\begin{align*}  
S^*_{\hat{\fgamma}} =  S_{\hat{\fgamma}} & -\Ima_{12}'\Ima_{22}^{-1} \fatS^{(k-1)}_{\hat{\fgamma}},\\  
\nu^*_{\hat{\fgamma},i} = \nu_{\hat{\fgamma},i}  & -\Ima_{12}'\Ima_{22}^{-1} \fnu^{(k-1)}_{\hat{\fgamma},i},   
\end{align*}
so that
$$S^*_{\hat{\fgamma}}=n^{-1/2} \sum_{i=1}^n  \nu^*_{\hat{\fgamma},i}.$$ 
Here
$$\Ima=  \begin{bmatrix}
      \hat{{ \mathcal{I} }}_{11} & \Ima_{12}' \\
     \Ima_{12} &  \Ima_{22}
\end{bmatrix},$$
with $\hat{{ \mathcal{I} }}_{11} \in \mathbb{R}$ and the  $(k-1) \times (k-1)$ matrix $\Ima_{22}$  assumed invertible, is a consistent estimate of the population Fisher information $\Imag$,  which is assumed to exist and  is the variance of $(\nu_{\fgamma_0,i}, {\fnu^{(k-1)}_{\fgamma_0,i}}')'$ marginal over $\bm{X}_i$, under $H_0$. The matrix  $\Imag$ is assumed to be continuous in the parameters.  In GLMs, typically $\Ima= n^{-1}\bm{X'}\hat{\bm{W}}\bm{X}$, where $\bm{X}$ is the design matrix and $\hat{\bm{W}}$ the estimated weight matrix \citep[][p. 126]{agresti2015foundations}.
Further, for $1\leq j \leq w$ we write
$$ S^{*j}_{\hat{\fgamma}} =  S^j_{\hat{\fgamma}}  -\Ima_{12}'\Ima_{22}^{-1} \fatS^{(k-1),j}_{\hat{\fgamma}}. $$

As discussed, $S_{\hat{\fgamma}}$ is not generally asymptotically equivalent to $S_{\fgamma_0}$.
The effective score $S^*_{\fgamma_0}$ (based on $\Ima=\Imag$) however is the residual from the projection  
of the score $S_{\fgamma_0}$ on the space spanned by the nuisance scores. Hence   $S^*_{\fgamma_0}$ is uncorrelated with the nuisance scores $\fatS^{(k-1)}_{{\fgamma_0}}$ \citep[][p. 344]{marohn2002comment}.
Correspondingly, as noted in the proof of Theorem \ref{signflipeff}, under mild regularity assumptions $S^*_{\hat{\fgamma}}= S^*_{\fgamma_0}+ o_{\mathbb{P}_{\beta_0,\fgamma_0}}(1)$, i.e., asymptotically the effective score  is not affected by the nuisance estimation.

Note that if $\hat{\fgamma}$ is the maximum likelihood estimate under $H_0$, then $\fatS^{(k-1)}_{\hat{\fgamma}}=\bm{0}$, so that $S^*_{\hat{\fgamma}}=  S_{\hat{\fgamma}}$. 
The summands $\nu^*_{\hat{\fgamma},i}$ and $\nu_{\hat{\fgamma},i}$ are different, however, 
and the key point is that  $S^*_{\hat{\fgamma}}= S^*_{\fgamma_0}+ o_{\mathbb{P}_{\beta_0,\fgamma_0}}(1)$.

Like \citet{marohn2002comment}, we assume that  if $\xi\in \mathbb{R}$ and $\beta=\beta^n=\beta_0+n^{-1/2}\xi$, then
$$S_{\hat{\fgamma}} = S_{\fgamma_0} - \Imag_{12}' \sqrt{n}(\hat{\fgamma}-\fgamma_0)+o_{\mathbb{P}_{ \beta^n,\fgamma_0}}(1), $$
$$ \fatS^{(k-1)}_{\hat{\fgamma}}= \fatS^{(k-1)}_{\fgamma_0} - \Imag_{22} \sqrt{n}(\hat{\fgamma}-\fgamma_0)+o_{\mathbb{P}_{ \beta^n,\fgamma_0}}(1),$$
which is satisfied under mild assumptions such as continuous second order derivatives.

\begin{theorem} \label{signflipeff}

Consider the test of Theorem \ref{mainST} with $T_j^n=S^{*j}_{\hat{\fgamma}}$, $1\leq j \leq w$.  
As $n\rightarrow \infty$,   under $H_0$ the rejection probability   converges to 
$\lfloor \alpha w \rfloor /w \leq \alpha$.
\end{theorem}

The test of Theorem \ref{signflipeff} has a parametric counterpart, which uses that under $H_0$, $S^*_{\hat{\fgamma}}$ is asymptotically normal with zero mean and known variance, the effective information \citep[][p. 341]{marohn2002comment}.
This test is asymptotically  equivalent to the test of Theorem \ref{signflipeff}, as the following proposition says.

\begin{proposition} \label{aseqeff}
Let $\xi\in \mathbb{R}$ and suppose the true parameter satisfies $\beta=\beta^n=\beta_0+n^{-1/2}\xi$.
As in Theorem \ref{signflipeff}, let $T_j^n=S^{*j}_{\hat{\fgamma}}$, $1\leq j \leq w$.
Define $\phi_{n,w}= \mathbbm{1}_{\{ T_1^n > T^n_{[1-\alpha]}\}}$  to be the test of Theorem \ref{signflipeff}.
Let $\phi'_n$ be the parametric test $\mathbbm{1}_{\{ T_1^n > \sigma_0 \Phi(1-\alpha) \}}$, where $\sigma_0^2\in\mathbb{R}$ is the effective Fisher information and $\Phi$  the cdf of the standard normal distribution.
Then $\lim_{w\rightarrow\infty}\liminf_{n\rightarrow\infty}\mathbb{P}(\phi_{n,w} =\phi'_n)=1.$
\end{proposition}

\subsection{Robustness} \label{robeff}

In Section \ref{secrob} it was explained that the test of Theorem \ref{mainST} is often robust against misspecification of the variance of the score. The test of Theorem \ref{signflipeff} is also robust against certain forms of variance misspecification. An example is the case that $S_{\hat{\fgamma}}$  and $\fatS^{(k-1)}_{\hat{\fgamma}}$ are  misspecified by the same factor, see Proposition \ref{missp}. This happens in particular if the variance is misspecified by a factor which is independent of the covariates.

\begin{proposition} \label{missp}
Suppose that $\Ima= n^{-1}\bm{X}'\hat{\bm{W}}\bm{X}$, where $\bm{X}$ is an $n\times k$ design matrix with i.i.d. rows and $\hat{\bm{W}}$ a weight matrix.
Consider a misspecification factor $c_1>0$ and misspecified scores
$$ \tilde{\nu}_{\hat{\fgamma},i}=c_1{\nu}_{\hat{\fgamma},i}, \quad  \tilde{\bm{\nu}}^{(k-1)}_{\hat{\fgamma},i}= 
c_1 \bm{\nu}^{(k-1)}_{\hat{\fgamma},i},\quad  i=1,...,n.$$
Further, for $c_2>0$ consider the misspecified weight matrix $\tilde{\bm{W}}=c_2 \hat{\bm{W}}.$
Let $\Imat= n^{-1} \bm{X}'\tilde{\bm{W}}\bm{X}$ be the misspecified average Fisher information.
Let $\tilde{\nu}^*_{\hat{\fgamma},i}=  \tilde{\nu}_{\hat{\fgamma},i}-\Imat_{12}'\Imat_{22}^{-1} \tilde{\bm{\nu}}^{(k-1)}_{\hat{\fgamma},i}$ be the misspecified effective scores, $i=1,...,n$.
Consider the test of Theorem \ref{signflipeff}, with  $S^{*j}_{\hat{\fgamma}}$, $j=1,...,w,$  replaced by  the misspecified effective score 
$$\tilde{S}^{*j}_{\hat{\fgamma}}= n^{-1/2} \sum_{i=1}^n g_{ji} \tilde{\nu}^*_{\hat{\fgamma},i}.$$ 
 Under $H_0$, as $n\rightarrow \infty$, the rejection probability  of this test converges to 
$\lfloor \alpha w \rfloor /w \leq \alpha$.
\end{proposition}

Proposition \ref{missp} is useful, since it tells us that if in a GLM $var(Y_i)$ is misspecified by a constant, such that $\hat{\bm{W}}$ and the scores are misspecified by a constant, the resulting test is still asymptotically exact.
In Proposition \ref{missp} we assume that the misspecification factors of the weights and the scores are the same for all observations.
This is satisfied for example when the model is binomial or Poisson, but the true distribution is respectively quasi-binomial or quasi-Poisson.
Moreover, in practice the test can be very robust against heteroscedasticity (see Section \ref{simulations}). The variance misspecification  is not generally allowed to depend on the covariates, since then   $S_{\hat{\fgamma}}$ and $\fatS^{(k-1)}_{\hat{\fgamma}}$ can be misspecified by different factors asymptotically. There are exceptions however, see Sections \ref{hetbf} and \ref{simulations}.

When there are estimated nuisance parameters, one can sometimes nevertheless decide to use the test of Theorem \ref{mainST} with the basic scores $\nu_{\hat{\gamma},i}$ plugged in (rather than using effective scores). Indeed, this test has been shown to be very robust to misspecification, as long as $\mathbb{E}\nu_{\hat{\gamma},i}=0$, $1\leq i \leq n$.
It is asymptotically conservative if  the score $S_{\fgamma_0}$ is correlated with the nuisance scores $\fatS^{(k-1)}_{\fgamma_0}$, i.e., when  $\Imag_{12}\neq 0$. Hence, when using this test, it can be useful to redefine the covariates  such that $\Imag_{12}=0$ \citep[as in][]{cox1987parameter}.
When  $\hat{\bm{W}}=b\bm{I}$, $b>0$, this means ensuring  that the nuisance covariates are orthogonal to the covariate of interest.
If the model is potentially misspecified, then the weights and hence $\Imag_{12}$ are not asymptotically known, but the user could substitute a best guess for the weights.

\subsection{An example} \label{hetbf}
As discussed,  the test of Theorem \ref{signflipeff} is not generally asymptotically exact if the variance misspecification depends on the covariates. An important exception is the case where 
the model is
\begin{equation} \label{simplelinmodel}
Y_i \sim N (\gamma_0+\beta X_i, \sigma^2) \quad i=1,...,n,
\end{equation}
where $\gamma_0$ is the unknown intercept and $X_i\in \mathbb{R}$ . If the null hypothesis is $H_0:\beta=\beta_0$, then $\gamma_0$ is a nuisance parameter that needs to be estimated. (We do not need to know $\sigma$ and can simply substitute $1$ for it.) Hence, we compute the effective score. Note that  for $1\leq i \leq n$,
\begin{align*}
\nu_{\hat{\fgamma},i} &=   {x_i(y_i-\hat{\mu}_i)}/ \sigma^{2},\\
\fnu^{(k-1)}_{\hat{\fgamma},i} & =   {(y_i-\hat{\mu}_i)}/\sigma^{2}.
\end{align*}
Note that we can consistently estimate 
$\Imag_{12}\Imag_{22}^{-1}$ by   $\overline{x} =\frac{1}{n} \sum_{i=1}^n x_i$, so that the effective score contributions are
$$\nu^*_{\hat{\fgamma},i}=  (x_i-\overline{x})  (y_i-\hat{\mu}_i)/\sigma^{2}.$$

Thus, the effective score contributions are exactly the basic score contributions after centering $x_1,...,x_n$ around $0$. Similarly, if  $x_1,...,x_n$ are already centered, then  $\nu_{\hat{\fgamma},i}$ and $\nu^*_{\hat{\fgamma},i}$  coincide, since then $\Ima_{12}=0$.

The test of Theorem \ref{signflipeff} is not always asymptotically exact if $S_{\hat{\fgamma}}$ and $\fatS^{(k-1)}_{\hat{\fgamma}}$ are misspecified by different factors. However, if $\Ima_{12}=0$, then this does not apply anymore. The test of Theorem \ref{signflipeff} then remains asymptotically exact and reduces to the test based on the basic score.
For the model \eqref{simplelinmodel}, this means that even if the misspecification of $var(Y_i)$ depends on $X_i$, we obtain an asymptotically exact test.

A particular case where this principle applies is the generalized Behrens-Fisher problem, where the aim is to test equality of the means $\mu^1$ and $\mu^2$ of two populations (or to test if  $\mu^1 \leq \mu^2$ or $\mu^1 \geq \mu^2$). 
In this problem, it is only assumed that two independent samples from these populations are available, without making other assumptions such as equal variances. 
It is well-known that this problem has no exact solution under normality \citep{pesarin2010permutation, lehmann2005testing}.
Under mild assumptions, we  obtain an asymptotically exact test for this problem. \citet{pesarin2010permutation} already suggested sign-flipping residuals to solve this problem. This is equivalent to flipping scores in our linear model \eqref{simplelinmodel} if $|x_1|=...=|x_n|$.

\section{Multi-dimensional parameter of interest} \label{secmv}

Until now  we have considered hypotheses about a one-dimensional parameter $\beta\in \mathbb{R}$. Here we extend our results to hypotheses about a multi-dimensional parameter $\fbeta\in \mathbb{R}^d$, $d\in\mathbb{N}$. 
Our tests are defined even if $d> n$, 
but in the theoretical results that follow we consider $d$ fixed and let $n$ increase to infinity. 
The extension to multi-dimensional $\fbeta$ shares important characteristics with the test for a one-dimensional parameter, such as robustness and asymptotic equivalence with the  parametric score test.

\subsection{Asymptotically exact test}

Our tests below are related to the existing nonparametric combination (NPC) methodology \citep{pesarin2001multivariate, pesarin2010permutation,pesarin2010finite}. This is a very general permutation-based methodology that allows combining test statistics for many hypotheses into a single test of the intersection hypothesis. NPC can be extended to the score-flipping framework. Our tests below could be considered a special case of such an extension of the NPC methodology.  This special case has certain power-optimality properties, discussed below.

The parametric score test  has a classical extension to a hypothesis on a multi-dimensional parameter, $H_0:\fbeta =\fbeta_0\in \mathbb{R}^d$ \citep{rao1948large}. 
We will extend our test in an analogous way. We first assume the nuisance $\fgamma_0\in \mathbb{R}^{k-d}$ to be known.
Since $\fbeta\in \mathbb{R}^d$, the score is $\fatS_{\fgamma_0}=n^{-1/2}\sum_{i=1}^n \fnu_{\fgamma_0,i}\in \mathbb{R}^d$, 
where
$$\fnu_{\fgamma_0,i}= \frac{\partial }{\partial \fbeta } \log f_{\fbeta,{\fgamma_0},\bm{X}_i}(Y_i)\vert_{\fbeta=\fbeta_0}\in \mathbb{R}^d,$$  $1\leq i \leq n$, 
which are now $d$-vectors.   We assume the derivatives exist. 
About the elements of $\fnu_{\fgamma,i}$ (and the nuisance scores considered later) we make the assumptions which are analogous to the earlier assumptions about  $\nu_{\fgamma,i}$.

Let $\Ima_{11}$ to be a consistent estimate of  $\Imag_{11}$, the $d\times d$ Fisher information for $\fbeta\in \mathbb{R}^d$. 
Rao's classical  statistic for testing $H_0:\fbeta =\fbeta_0 \in \mathbb{R}^d$ is 
$$
\fatS_{\fgamma_0}'\Ima_{11}^{-1}\fatS_{\fgamma_0}= \Big ( n^{-1/2}\sum_{i=1}^n \fnu'_{\fgamma_0,i} \Big ) \Ima_{11}^{-1} \Big (  n^{-1/2}\sum_{i=1}^n \fnu_{\fgamma_0,i}\Big ),
$$
which asymptotically has a $\chi^2_d$ distribution under $H_0$.

Instead of requiring a matrix $\Ima^{-1}$ which converges to the inverse of the Fisher Information, in our test that follows we allow to replace the Fisher information by any random matrix  $\hat{\bm{V}}$ converging to some  non-zero matrix $ \bm{V}$. That is, we do not require the Fisher Information to be asymptotically known, just like in the one-dimensional case.
The matrix $\bm{V}$ can be any matrix of preference, including $\Imag_{11}^{-1}$ (if $\Imag_{11}$ is invertible), or we can take $\hat{\bm{V}}=\bm{V}=\bm{I}$. We will discuss different choices of  $\bm{V}$ shortly.

\begin{theorem} \label{mvknownn}
The result of Theorem \ref{mainST} still applies if for $1\leq j \leq w$ we define
$$
T_j^n=  
\Big (n^{-1/2}\sum_{i=1}^n g_{ji}\fnu'_{\fgamma_0,i}  \Big ) \hat{\bm{V}}   \Big ( n^{-1/2}\sum_{i=1}^n g_{ji}\fnu_{\fgamma_0,i}  \Big ).
$$
\end{theorem}

In case the nuisance parameter $\fgamma_0$ is unknown and we have a $\sqrt{n}$-consistent estimate $\hat{\fgamma}$, we can use the same test, but with effective scores instead of  basic scores plugged in. See Theorem \ref{mvunknownn}.
For multi-dimensional $\fbeta$, the effective score contributions are 
 $$\fnu^*_{\hat{\fgamma},i} = \fnu_{\hat{\fgamma},i}  -\Ima_{12}'\Ima_{22}^{-1} \fnu^{(k-d)}_{\hat{\fgamma},i}\in \mathbb{R}^d,$$
 $1\leq i \leq n$, where  $$\fnu^{(k-d)}_{\hat{\fgamma},i}=
\frac{\partial }{\partial \fgamma} \log f_{\fbeta_0,{\fgamma},\bm{X}_i}({Y_i}) \Big\vert_{\fgamma=\hat{\fgamma}}\in\mathbb{R}^{k-d}.$$
Here $\Ima_{12}$ and $\Ima_{22}$ are $(k-d)\times d$ and $(k-d)\times (k-d)$ matrices, respectively.

\begin{theorem}[Unknown nuisance] \label{mvunknownn}
The result of Theorem \ref{mainST} still applies if for $1\leq j \leq w$ we define
$$
T_j^n=  \Big (n^{-1/2}\sum_{i=1}^n g_{ji}\fnu_{\hat{\fgamma},i}^{*'}  \Big ) \hat{\bm{V}}   \Big ( n^{-1/2}\sum_{i=1}^n g_{ji}\fnu^*_{\hat{\fgamma},i}  \Big ).
$$
\end{theorem}

The  test of Theorem \ref{mvunknownn} is asymptotically equivalent to its parametric counterpart, as Proposition \ref{aseqmv} states. In particular, if we take $\hat{\bm{V}}=(\Ima^{*})^{-1}$, where $(\Ima^{*})^{-1}$ is a consistent estimate of the inverse of the effective Fisher information, then the test of Theorem \ref{mvunknownn} is asymptotically equivalent to the parametric score test \citep[][p. 86]{hall1990large}.

\begin{proposition}[Equivalence with parametric counterpart] \label{aseqmv}
Define $T_j^n$ as in Theorem \ref{mvunknownn}, $1\leq j \leq w$.
Let $\bm{\xi}\in \mathbb{R}^d$ and suppose the true value of the parameter of interest is $\bm{\beta}=\bm{\beta}^n=\bm{\beta}_0+n^{-1/2}\bm{\xi}$.
Let $\phi_{n,w}=\mathbbm{1}_{\{T_1^n>T_{[1-\alpha]}^n\}}$. This is the test of Theorem \ref{mvunknownn}.
Let   $\phi'_{n}$ be the parametric  test $\mathbbm{1}_{\{T_1^n>q_{\alpha}\}}$, where $q_{\alpha}$ is the $(1-\alpha)$-quantile  of the  distribution to which $T_1^n$ converges as $n\rightarrow\infty$  under $\bm{\beta}=\bm{\beta}_0$. (This is the $\chi^2_d$ distribution if $\bm{V}$ is the inverse of the effective information matrix $\Imag^*$).
Then $\lim_{w\rightarrow\infty}\liminf_{n\rightarrow\infty}\mathbb{P}(\phi_{n,w} =\phi'_n)=1.$
\end{proposition}

We have seen that the test of Theorem \ref{mainST} is often robust against overdispersion and heteroscedasticity: as long as the score contributions have mean $0$, the test is asymptotically exact, under very mild assumptions. Moreover, it is not required to estimate the Fisher information. The  same applies to the multi-dimensional extension in Theorem \ref{mvknownn}. 

The test that  takes into account nuisance estimation (Theorem \ref{mvunknownn}) uses effective scores, so that it does require  estimating the  information.  However, as in the one-dimensional case, it can be seen that the test remains valid if the information matrix  is asymptotically misspecified by a constant (as in Proposition \ref{missp}). Additional robustness is illustrated with simulations in Section \ref{simmv}.

\subsection{Connection with the global test}

The test of Theorem \ref{mvknownn} is related to the global test, which was developed in \citet{goeman2004global,goeman2006testing,goeman2011testing}. We can combine the global test with the score-flipping approach. In certain cases, the resulting test coincides with the test of Theorem \ref{mvknownn}.

The global test is a parametric test of $H_0$.  For the test to be defined, it is not required that $d\leq n$.  For GLMs with canonical link function, the test statistic of the global test is  
\begin{equation} \label{gttt}
\fatS_{\fgamma_0}' \bm{\Sigma} \fatS_{\fgamma_0},
\end{equation}
with $\bm{\Sigma}$ a freely chosen positive (semi)definite $d \times d$ matrix \citep{goeman2006testing,goeman2011testing}. The choice of $\bm{\Sigma}$ influences the power properties.
 
 Note that when $\hat{\bm{V}} = \bm{\Sigma}$, the statistic \eqref{gttt}  coincides with the statistic of  Theorem \ref{mvknownn}. Thus, it immediately follows from our results that the global test can be combined with our sign-flipping approach, leading to a test which becomes asymptotically exact  as $n\rightarrow\infty$  and asymptotically equivalent to its parametric counterpart, the original global test (by Proposition \ref{aseqmv}). Combining the global test with sign-flipping is useful in the light of our robustness results.  Moreover, the sign-flipping variant  can be combined with existing permutation-based multiple testing methodology \citep{westfall1993resampling, hemerik2018false,hemerik2019permutation}. 
 
\citet{goeman2006testing} provide results on the power properties of the global test as depending on the choice of $\bm{\Sigma}$. Since the global test is asymptotically equivalent to its sign-flipping counterpart, these results can be used as recommendations on the choice of $\hat{\bm{V}}$ in Theorem \ref{mvknownn}. In particular, according to \citet[][Section 8]{goeman2006testing}, taking $\hat{\bm{V}}=\bm{I}$ leads to good power if one expects that relatively much of the variance of $\bm{Y}$ is explained by the large variance principal components of the design matrix. If this is not the case, taking $\hat{\bm{V}}$ to be an estimate of the inverse of the Fisher information (if invertible) can provide better power. 
In general, the global  test has optimal power on average (over $\fbeta$) in a neighbourhood of $\fbeta_0$ that depends on $\bm{\Sigma}$ \citep{goeman2006testing}. Hence the same holds asymptotically for the test of Theorem \ref{mvknownn}, for GLMs with canonical link.

\section{Simulations} \label{simulations}

To compare the tests in this paper with each other and existing tests, we applied them to simulated data. In particular we considered scenarios where the  model was misspecified.  Simulations with a multi-dimensional parameter of interest are in  Section \ref{simmv}.

\subsection{Overdispersion, heteroscedasticity and  estimated nuisance} \label{overdEstnuis}

In Sections  \ref{overdEstnuis} and \ref{secignn} the assumed model was Poisson, but in fact $Y_1,...,Y_n$ were drawn from a negative binomial distribution. 

The covariates $X, Z,Z^l \in \mathbb{R}$ were drawn from a multivariate normal distribution with zero mean and $var(X)=var(Z)=var(Z^l)=1.$ (For nonzero means, similar simulation results were obtained as below.)
The response satisfied 
$\log\{\mathbb{E}(Y_i)\}=\log(\mu_i) =\lp_i= $
$$0+ \beta \cdot  X_i + \gamma_0 \cdot  Z_i + \gamma_0^l \cdot Z_i^l.$$

The null hypothesis was $H_0:\beta=0$.
In Section \ref{overdEstnuis} we took $\gamma_0^l=0$.  The coefficient $\gamma_0$ and the intercept $0$   were nuisance parameters that were estimated by maximum likelihood under $H_0$. We took $\gamma_0=1$ and $\rho(X_i,Z_i)=0.5, $ $\rho(Z_i^l,Z_i)=0$, $\rho(Z_i^l,X_i)=0$.
We took the dispersion parameter of the negative binomial distribution to be $1$, so that $var(Y_i)=\mu_i+\mu_i^2$. 

The assumed model, however, was Poisson, i.e., $var(Y_i) =\mu_i$ was assumed. 
Thus the true variance was larger than the assumed variance and  the variance misspecification factor depended on $\mu_i$, i.e., on the covariate $Z_i$. The assumed log link function was correct and in Section \ref{overdEstnuis} the linear predictor was correct as well.

In Figure \ref{rejproboverdisp} the estimated rejection probabilites of four tests under ${H_0:\beta=0}$ are compared, based on 5000 repeated simulations. In all simulations the tests were two-sided. 

One of the tests considered was the parametric score test. Since the assumed model was Poisson,   the computed Fisher information was too small and the test was anti-conservative.

We also applied a Wald test, where we used a sandwich estimate \citep[][p. 280]{agresti2015foundations} of the variance of $\hat{\beta}$, to correct for the misspecified variance function. We used the R package \emph{gee} for this (available on CRAN), specifying blocks of size $1$. As can be seen in Figure \ref{rejproboverdisp}, this test was rather anti-conservative (especially for small $\alpha$, e.g. $\alpha=0.01$). This was in particular due to the estimation error of the sandwich \citep{boos1992on,freedman2006so,maas2004robustness,kauermann2000sandwich}.

Further, we applied the  sign-flipping test based on the basic scores $\nu_{\hat{\gamma},i}$.   Due to the estimation of $\gamma_0$, the variance of the score was shrunk and the test was conservative, as explained in Section \ref{seceffsctest}.  In the simulations under $H_0$ we took $w=200$. Taking $w$ larger led to a very similar level (see also \citeauthor{marriott1979barnard}, \citeyear{marriott1979barnard}). In the power simulations we took $w=1000.$

Finally, we used the sign-flipping test of Theorem \ref{signflipeff}, which is based on the effective scores $\nu^*_{\hat{\gamma},i}$. 
In Section \ref{robeff} it was already  shown that this test is asymptotically exact under constant variance misspecification. Here, however, the variance misspecification factor was $1+\mu_i$ (i.e., it depended on $Z_i$). Nevertheless the rejection probability under $H_0$ was approximately $\alpha$. This illustrates that the test has some additional robustness, which we have not theoretically shown.


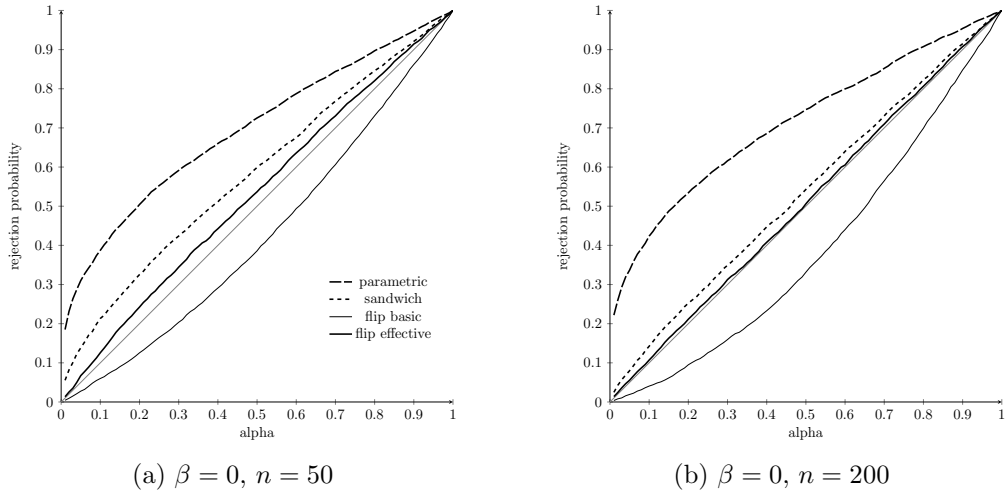
\begin{figure}
\centering
\begin{subfigure}{.5\textwidth}

\centering
\begin{tikzpicture}[scale=0.5]

\begin{axis}[
    anchor=origin,
                legend style={at={(7cm,2.5cm)}, anchor=west, legend columns=1, draw=none},
    width=12cm, height=12cm,
axis x line=bottom, xlabel={alpha}, xmin=0, xmax=1,  
       axis y line=left, ylabel={rejection probability}, ymin=0, ymax=1]

\draw[gray] (0,0) -- (12cm,12cm);

\addplot[dash pattern=on 9pt off 2pt , line width = 1, very thick,samples=500] table[x="alpha",y="par", col sep=comma,smooth]{rejProbs_gamma1_dispexp0_n50.csv}; 

\addplot[ dashed, line width = 1, very thick] table[x="alpha",y="GEE", col sep=comma]{rejProbs_gamma1_dispexp0_n50.csv};

\addplot[ line width = 0.5] table[x="alpha",y="flipSimple", col sep=comma, smooth]{rejProbs_gamma1_dispexp0_n50.csv};  

\addplot[ line width = 1, very thick  ] table[x="alpha",y="flipEff", col sep=comma, smooth]{rejProbs_gamma1_dispexp0_n50.csv};  

\legend{parametric, sandwich, flip basic, flip effective};
\end{axis}
\end{tikzpicture}

\caption{$\beta=0$, $n=50$}
\label{rejproboverdisp2a}

\end{subfigure}
\begin{subfigure}{.5\textwidth}

\centering
\begin{tikzpicture}[scale=0.5]

\begin{axis}[
    anchor=origin,
                legend style={at={(12cm,.7)}, anchor=west, legend columns=1, draw=none},
    width=12cm, height=12cm,
axis x line=bottom, xlabel={alpha}, xmin=0, xmax=1,  
       axis y line=left, ylabel={rejection probability}, ymin=0, ymax=1]

\draw[gray] (0,0) -- (12cm,12cm);

\addplot[dash pattern=on 9pt off 2pt , line width = 1, very thick,samples=500] table[x="alpha",y="par", col sep=comma,smooth]{rejProbs_gamma1_dispexp0_n200.csv}; 

\addplot[ dashed, line width = 1, very thick] table[x="alpha",y="GEE", col sep=comma]{rejProbs_gamma1_dispexp0_n200.csv};

\addplot[ line width = 0.5] table[x="alpha",y="flipSimple", col sep=comma, smooth]{rejProbs_gamma1_dispexp0_n200.csv};  

\addplot[ line width = 1, very thick  ] table[x="alpha",y="flipEff", col sep=comma, smooth]{rejProbs_gamma1_dispexp0_n200.csv};

\end{axis}
\end{tikzpicture}

\caption{$\beta=0$, $n=200$}
\label{rejproboverdisp2b}

\end{subfigure}
\caption{Estimated rejection probabilities for four tests under misspecified variance and estimated nuisance. The null hypothesis was $H_0:\beta=0$. }
\label{rejproboverdisp}
\end{figure}


\subsection{Ignored nuisance} \label{secignn}
The same simulations  were performed as in Section \ref{overdEstnuis}, but with $\gamma_0^l=1$. Since $\gamma_0^l=0$ was assumed,  $Z_i^l$ represented an ignored, latent variable. Figure  \ref{ignn} shows similar results as Figure \ref{rejproboverdisp}. The parametric test was even more anti-conservative than in Section \ref{overdEstnuis}. The reason is that the introduction of $Z_i^l$ increased the variance  $Y_i$, so that the variance of the score was even more misspecified than in Section \ref{overdEstnuis}.

The  test of Theorem \ref{signflipeff} was still nearly exact for $n=200$, even though $\mu_i$ was misspecified. (Even marginally over $Z_i^l$, $\mu_i$ was misspecified. Possibly the estimation of the intercept corrected for the misspecification.)

A conclusion from the simulations of Sections \ref{overdEstnuis} and \ref{secignn}, is that  the sandwich-based approach should not always be seen as the most reliable way of testing models with misspecified variance functions. Indeed, in our simulations the test of Theorem \ref{signflipeff} was substantially less anti-conservative (while having similar power, see Section \ref{powerflip}).

\begin{figure}
\centering
\begin{subfigure}{.5\textwidth}

\centering
\begin{tikzpicture}[scale=0.5]

\begin{axis}[
    anchor=origin,
                legend style={at={(7cm,2.5cm)}, anchor=west, legend columns=1, draw=none},
    width=12cm, height=12cm,
axis x line=bottom, xlabel={alpha}, xmin=0, xmax=1,  
       axis y line=left, ylabel={rejection probability}, ymin=0, ymax=1]

\draw[gray] (0,0) -- (12cm,12cm);

\addplot[dash pattern=on 9pt off 2pt , line width = 1, very thick,samples=500] table[x="alpha",y="par", col sep=comma,smooth]{rejProbs_g1_l1_rhoxz05_n50.csv}; 

\addplot[ dashed, line width = 1, very thick] table[x="alpha",y="GEE", col sep=comma]{rejProbs_g1_l1_rhoxz05_n50.csv};

\addplot[ line width = 0.5] table[x="alpha",y="flipSimple", col sep=comma, smooth]{rejProbs_g1_l1_rhoxz05_n50.csv};  

\addplot[ line width = 1, very thick  ] table[x="alpha",y="flipEff", col sep=comma, smooth]{rejProbs_g1_l1_rhoxz05_n50.csv};  

\legend{parametric, sandwich, flip basic, flip effective};
\end{axis}
\end{tikzpicture}

\caption{$n=50,\beta=0$}
\label{ignna}

\end{subfigure}
\begin{subfigure}{.5\textwidth}

\centering
\begin{tikzpicture}[scale=0.5]

\begin{axis}[
    anchor=origin,
                legend style={at={(12cm,.7)}, anchor=west, legend columns=1, draw=none},
    width=12cm, height=12cm,
axis x line=bottom, xlabel={alpha}, xmin=0, xmax=1,  
       axis y line=left, ylabel={rejection probability}, ymin=0, ymax=1]

\draw[gray] (0,0) -- (12cm,12cm);

\addplot[dash pattern=on 9pt off 2pt , line width = 1, very thick,samples=500] table[x="alpha",y="par", col sep=comma,smooth]{rejProbs_g1_l1_rhoxz05_n200.csv}; 

\addplot[ dashed, line width = 1, very thick] table[x="alpha",y="GEE", col sep=comma]{rejProbs_g1_l1_rhoxz05_n200.csv};

\addplot[ line width = 0.5] table[x="alpha",y="flipSimple", col sep=comma, smooth]{rejProbs_g1_l1_rhoxz05_n200.csv};  

\addplot[ line width = 1, very thick  ] table[x="alpha",y="flipEff", col sep=comma, smooth]{rejProbs_g1_l1_rhoxz05_n200.csv};

\end{axis}
\end{tikzpicture}

\caption{$n=200,\beta=0$}
\label{ignnb}

\end{subfigure}
\caption{Estimated rejection probabilities for four tests under misspecified variance, estimated nuisance and ignored nuisance. The null hypothesis was $H_0:\beta=0$.}
\label{ignn}
\end{figure}
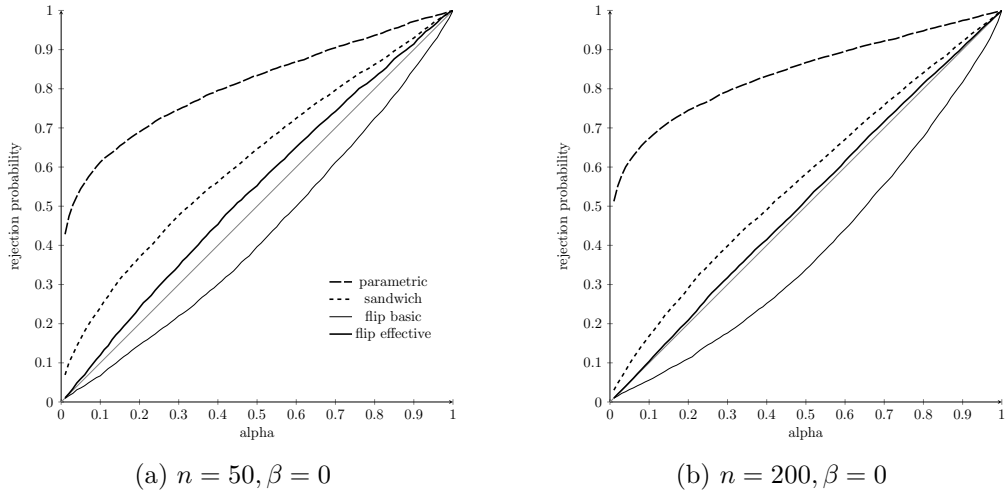

\subsection{Power} \label{powerflip}

For a meaningful power comparison of the four tests, we considered the scenario where the assumed model was correct, i.e.,
 the data distribution was Poisson and $\gamma_0^l$ was $0$.
See figure \ref{powerb02}. The estimated probabilities are based on $2\cdot 10^4$ simulation loops. 

Since  the model was correct, asymptotically there was no better choice than the parametric test. The sign-flipping test of Theorem \ref{signflipeff} had very similar power. 
The basic sign-flipping test was again conservative due to the estimation of $\gamma_0$.
The sandwich-based test had the most power, but  was anti-conservative (null behavior not shown).

\begin{figure}
\centering
\begin{subfigure}{.5\textwidth}

\centering
\begin{tikzpicture}[scale=0.5]

\begin{axis}[
    anchor=origin,
                legend style={at={(7cm,2.5cm)}, anchor=west, legend columns=1, draw=none},
    width=12cm, height=12cm,
axis x line=bottom, xlabel={alpha}, xmin=0, xmax=1,  
       axis y line=left, ylabel={rejection probability}, ymin=0, ymax=1]

\draw[gray] (0,0) -- (12cm,12cm);

\addplot[dash pattern=on 9pt off 2pt , line width = 1, very thick,samples=500] table[x="alpha",y="par", col sep=comma,smooth]{power1.csv};    

\addplot[ dashed, line width = 1, very thick] table[x="alpha",y="GEE", col sep=comma]{power1.csv};

\addplot[ line width = 0.5] table[x="alpha",y="flipSimple", col sep=comma, smooth]{power1.csv};  

\addplot[ line width = 1, very thick  ] table[x="alpha",y="flipEff", col sep=comma, smooth]{power1.csv};  

\legend{parametric, sandwich, flip basic, flip effective};
\end{axis}
\end{tikzpicture}

\caption{$\beta=0.2$, $n=50$}
\label{powerb02a}

\end{subfigure}
\begin{subfigure}{.5\textwidth}

\centering
\begin{tikzpicture}[scale=0.5]

\begin{axis}[
    anchor=origin,
                legend style={at={(12cm,.7)}, anchor=west, legend columns=1, draw=none},
    width=12cm, height=12cm,
axis x line=bottom, xlabel={alpha}, xmin=0, xmax=1,  
       axis y line=left, ylabel={rejection probability}, ymin=0, ymax=1]

\draw[gray] (0,0) -- (12cm,12cm);

\addplot[dash pattern=on 9pt off 2pt , line width = 1, very thick,samples=500] table[x="alpha",y="par", col sep=comma,smooth]{power2.csv}; 

\addplot[ dashed, line width = 1, very thick] table[x="alpha",y="GEE", col sep=comma]{power2.csv};

\addplot[ line width = 0.5] table[x="alpha",y="flipSimple", col sep=comma, smooth]{power2.csv};  

\addplot[ line width = 1, very thick  ] table[x="alpha",y="flipEff", col sep=comma, smooth]{power2.csv};

\end{axis}
\end{tikzpicture}

\caption{$\beta=0.2$, $n=200$}
\label{powerb02b}

\end{subfigure}
\caption{Power comparison of four two-sided tests under the correct model, with estimated nuisance. The null hypothesis was $H_0:\beta=0$.}
\label{powerb02}
\end{figure}
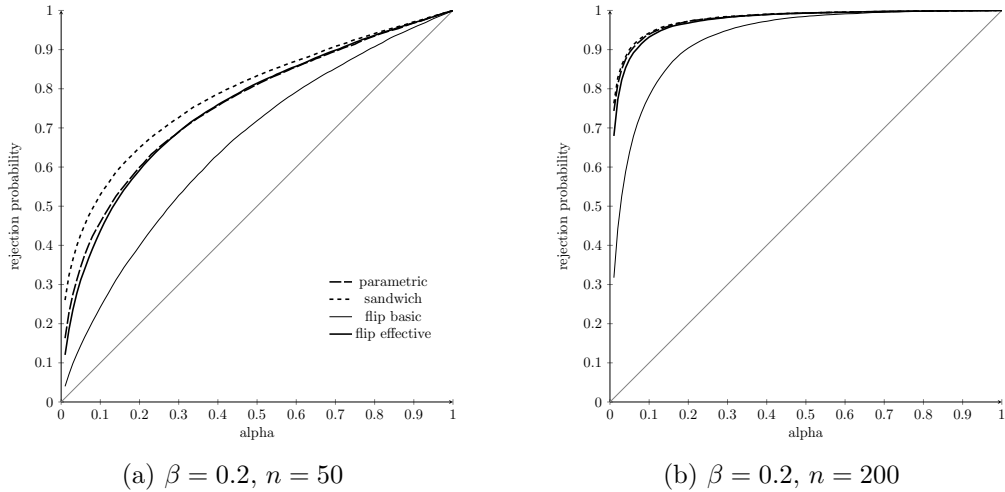

\subsection{Strong heteroscedasticity}
When a   Gaussian linear model is considered with $Y_i\sim N(\beta x_i,\sigma^2)$, $x_1=...=x_n=1$ and $H_0:\beta=0$, the score contributions are $\nu_i=X_i (Y_i-0)/\sigma^2=Y_i/\sigma^2 $. Thus the test of Theorem \ref{mainST} simply flips the observations $Y_i$, $1\leq i \leq n$. The parametric counterpart of this test is the one-sample t-test. 
The t-test needs to explicitly estimate the nuisance parameter $\sigma^2$;
 the sign-flipping test does not (simply substitute $\sigma=1$).

We simulated strongly heteroscedastic data:
 we took $Y_i\sim N(\beta x_i,\sigma_i^2)$, with $\sigma_i=\exp(i)$, $1\leq i \leq n=10.$
Consequently the t-statistic did not have the assumed distribution and  under $H_0$ the rejection probability  of the t-test was  far from the nominal level for most $\alpha$, see Figure \ref{heta}. The sign-flipping test did not need to estimate the variance. In this setting the test has  rejection probability  $\lfloor \alpha w \rfloor /w$ exactly if the transformations $g_1,...,g_w$ are drawn without replacement, since the observations are symmetric, see Proposition \ref{exacttest}. (We drew $g_1,...,g_w$ with replacement for convenience, but this gives almost the same test as drawing without replacement, due to the small probability of ties.)

For a meaningful power comparison, we considered the correct, homoscedastic model   with $\sigma_1=...=\sigma_{10}=1$.
Figure \ref{hetb}, based on $10^5$ repeated simulations, shows that the tests had virtually the same power.

\begin{figure}
\centering
\begin{subfigure}{.5\textwidth}

\centering

\begin{tikzpicture}[scale=0.5]

\begin{axis}[
    anchor=origin,
                legend style={at={(7cm,2.5cm)}, anchor=west, legend columns=1, draw=none},
    width=12cm, height=12cm,
axis x line=bottom, xlabel={alpha}, xmin=0, xmax=1,  
       axis y line=left, ylabel={rejection probability}, ymin=0, ymax=1]

\draw[gray] (0,0) -- (12cm,12cm);

\addplot[dash pattern=on 9pt off 2pt , line width = 1, very thick,samples=500] table[x="alpha",y="Parametric", col sep=comma,smooth]{rejProbs_ttest_H0_disp1_n10.csv}; 

\addplot[ dashed, line width = 1, very thick] table[x="alpha",y="Flip test", col sep=comma]{rejProbs_ttest_H0_disp1_n10.csv};

\legend{parametric test, flip test};
\end{axis}
\end{tikzpicture}

\caption{$\mu=0$, strong heteroscedasticity}

\label{heta}

\end{subfigure}
\begin{subfigure}{.5\textwidth}

\centering

\begin{tikzpicture}[scale=0.5]

\begin{axis}[
    anchor=origin,
                legend style={at={(7cm,2.5cm)}, anchor=west, legend columns=1, draw=none},
    width=12cm, height=12cm,
axis x line=bottom, xlabel={alpha}, xmin=0, xmax=1,  
       axis y line=left, ylabel={rejection probability}, ymin=0, ymax=1]

\draw[gray] (0,0) -- (12cm,12cm);

\addplot[dash pattern=on 9pt off 2pt , line width = 1, very thick,samples=500] table[x="alpha",y="Parametric", col sep=comma,smooth]{rejProbs_ttest_mu05_disp0_n10.csv}; 

\addplot[ dashed, line width = 1, very thick] table[x="alpha",y="Flip test", col sep=comma]{rejProbs_ttest_mu05_disp0_n10.csv};

\end{axis}
\end{tikzpicture}

\caption{$\mu=0.5$, correct model }

\label{hetb}

\end{subfigure}
\caption{Comparison of the one-sample t-test and the sign-flipping test. The null hypothesis was $H_0:\mu=0$.}
\label{het}
\end{figure}
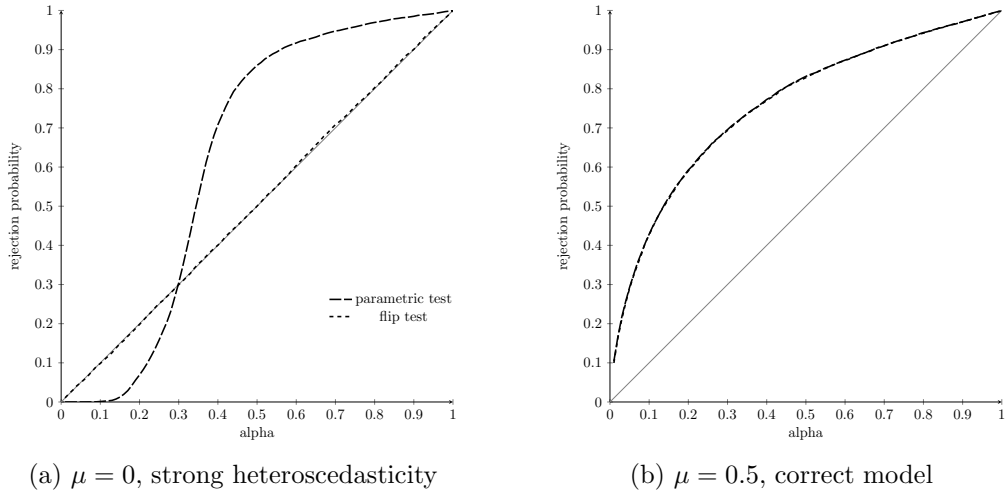

\subsection{Multi-dimensional parameter of interest} \label{simmv}

We considered the same setting as in Section \ref{secignn}, except that $\fbeta$ and the estimated nuisance parameter $\fgamma_0=(0.5,0.2,0,0,0)$ were $5$-dimensional (so $\bm{X}_i$, $\bm{Z}_i\in \mathbb{R}^5$). All corresponding covariates were correlated ($\rho=0.5$). 
There was an ignored nuisance covariate as before ($\gamma_0^l=0.5$), which was uncorrelated with the other covariates.  Thus there were in total 11 covariates. We took the overdispersion such that  $var(Y_i)=\mu_i+0.5 \mu_i^2$, i.e., the overdispersion again depended on the covariates (heteroscedasticity).

Instead of the basic score test  for a one-dimensional parameter we now used the multi-dimensional extension in Theorem \ref{mvknownn}. Similarly, instead of the test of Theorem \ref{signflipeff} based on effective scores, we used the multi-dimensional extension in Theorem \ref{mvunknownn}. We took $\hat{\bm{V}}=\bm{V}$ to be the identity matrix.

In Sections \ref{overdEstnuis} and \ref{secignn}  we compared our tests with a Wald test based on a sandwich estimate of $Var(\hat{\fbeta})$. Here we proceeded analogously, using  a sandwich estimate of the $5\times 5$  matrix $Var(\hat{\fbeta})$ in the multi-dimensional Wald test. This test uses that $\hat{\fbeta}'Var(\hat{\fbeta})^{-1}\hat{\fbeta}$ asymptotically has a $\chi^2_d$ distribution under the null hypothesis $H_0: \fbeta=\bm{0}$.

The results under $H_0$ are shown in Figure \ref{rejProbsMV}, where each plot is based on $10^4$ simulation loops. They are comparable to those in Section \ref{secignn}, except that the sandwich-based method is now even more anti-conservative. This is because $Var(\hat{\fbeta})$ is now a $5 \times 5$ matrix, which is difficult to estimate accurately. For $n=50$ and $\alpha=0.01$, the rejection probability of the sandwich-based method was $0.27$ instead of the required $0.01$.

For a meaningful power comparison of the four tests, we again considered the scenario where the assumed model was correct, i.e., the data distribution was Poisson and $\gamma_0^l$ was $0$. See Figure \ref{powerMV}, where each plot is based on $10^4$ simulation loops.   As usual, the sign-flipping test based on basic scores had low power due to nuisance estimation. 
The power of the sign-flipping test based on effective scores was comparable to that of the parametric score test. 
As in Section \ref{powerflip}, the test based on a sandwich estimate was the most powerful, but this has limited meaning, since it was also rather anti-conservative under the correct model (plot not shown). 

To conclude, sign-flipping provided much more reliable type-I error control than the sandwich approach, while giving satisfactory power (comparable to that of the parametric test, under the correct model).


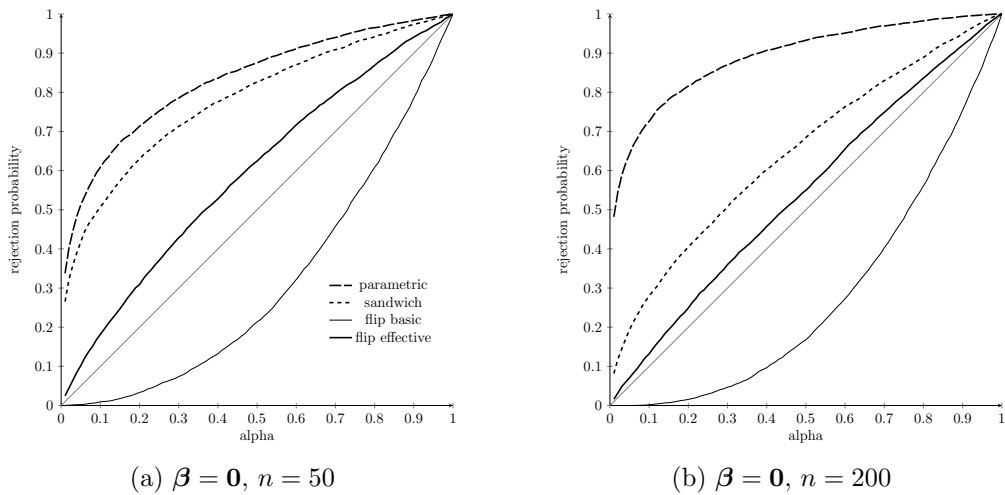
\begin{figure} 
\centering
\begin{subfigure}{.5\textwidth}

\centering
\begin{tikzpicture}[scale=0.5]

\begin{axis}[
    anchor=origin,
                legend style={at={(7cm,2.5cm)}, anchor=west, legend columns=1, draw=none},
    width=12cm, height=12cm,
axis x line=bottom, xlabel={alpha}, xmin=0, xmax=1,  
       axis y line=left, ylabel={rejection probability}, ymin=0, ymax=1]

\draw[gray] (0,0) -- (12cm,12cm);

\addplot[dash pattern=on 9pt off 2pt , line width = 1, very thick,samples=500] table[x="alpha",y="par", col sep=comma,smooth]{rejProbs1.csv}; 

\addplot[ dashed, line width = 1, very thick] table[x="alpha",y="GEE", col sep=comma]{rejProbs1.csv};

\addplot[ line width = 0.5] table[x="alpha",y="flipSimple", col sep=comma, smooth]{rejProbs1.csv};  

\addplot[ line width = 1, very thick  ] table[x="alpha",y="flipEff", col sep=comma, smooth]{rejProbs1.csv};  

\legend{parametric, sandwich, flip basic, flip effective};
\end{axis}
\end{tikzpicture}

\caption{$\fbeta=\bm{0}$, $n=50$}
\label{rejProbsMVa}

\end{subfigure}
\begin{subfigure}{.5\textwidth}

\centering
\begin{tikzpicture}[scale=0.5]

\begin{axis}[
    anchor=origin,
                legend style={at={(12cm,.7)}, anchor=west, legend columns=1, draw=none},
    width=12cm, height=12cm,
axis x line=bottom, xlabel={alpha}, xmin=0, xmax=1,  
       axis y line=left, ylabel={rejection probability}, ymin=0, ymax=1]

\draw[gray] (0,0) -- (12cm,12cm);

\addplot[dash pattern=on 9pt off 2pt , line width = 1, very thick,samples=500] table[x="alpha",y="par", col sep=comma,smooth]{rejProbs2.csv}; 

\addplot[ dashed, line width = 1, very thick] table[x="alpha",y="GEE", col sep=comma]{rejProbs2.csv};

\addplot[ line width = 0.5] table[x="alpha",y="flipSimple", col sep=comma, smooth]{rejProbs2.csv};  

\addplot[ line width = 1, very thick  ] table[x="alpha",y="flipEff", col sep=comma, smooth]{rejProbs2.csv};

\end{axis}
\end{tikzpicture}

\caption{$\fbeta=\bm{0}$, $n=200$}
\label{rejProbsMVb}

\end{subfigure}
\caption{ Estimated rejection probabilities under the null hypothesis. The model was misspecified due to overdispersion, heteroscedasticity and ignored nuisance.}
\label{rejProbsMV}
\end{figure}


\begin{figure} 
\centering
\begin{subfigure}{.5\textwidth}

\centering
\begin{tikzpicture}[scale=0.5]

\begin{axis}[
    anchor=origin,
                legend style={at={(7cm,2.5cm)}, anchor=west, legend columns=1, draw=none},
    width=12cm, height=12cm,
axis x line=bottom, xlabel={alpha}, xmin=0, xmax=1,  
       axis y line=left, ylabel={rejection probability}, ymin=0, ymax=1]

\draw[gray] (0,0) -- (12cm,12cm);

\addplot[dash pattern=on 9pt off 2pt , line width = 1, very thick,samples=500] table[x="alpha",y="par", col sep=comma,smooth]{power3.csv};   

\addplot[ dashed, line width = 1, very thick] table[x="alpha",y="GEE", col sep=comma]{power3.csv};

\addplot[ line width = 0.5] table[x="alpha",y="flipSimple", col sep=comma, smooth]{power3.csv};  

\addplot[ line width = 1, very thick  ] table[x="alpha",y="flipEff", col sep=comma, smooth]{power3.csv};  

\legend{parametric, sandwich, flip basic, flip effective};
\end{axis}
\end{tikzpicture}

\caption{$\fbeta=(0.2,0,0,0,0)'$, $n=50$}
\label{powerMVa}

\end{subfigure}
\begin{subfigure}{.5\textwidth}

\centering
\begin{tikzpicture}[scale=0.5]

\begin{axis}[
    anchor=origin,
                legend style={at={(12cm,.7)}, anchor=west, legend columns=1, draw=none},
    width=12cm, height=12cm,
axis x line=bottom, xlabel={alpha}, xmin=0, xmax=1,  
       axis y line=left, ylabel={rejection probability}, ymin=0, ymax=1]

\draw[gray] (0,0) -- (12cm,12cm);

\addplot[dash pattern=on 9pt off 2pt , line width = 1, very thick,samples=500] table[x="alpha",y="par", col sep=comma,smooth]{power4.csv};  

\addplot[ dashed, line width = 1, very thick] table[x="alpha",y="GEE", col sep=comma]{power4.csv};

\addplot[ line width = 0.5] table[x="alpha",y="flipSimple", col sep=comma, smooth]{power4.csv};  

\addplot[ line width = 1, very thick  ] table[x="alpha",y="flipEff", col sep=comma, smooth]{power4.csv};

\end{axis}
\end{tikzpicture}

\caption{$\fbeta=(0.2,0,0,0,0)'$, $n=200$}
\label{powerMVb}

\end{subfigure}
\caption{ Power comparison under the correct model. The null hypothesis was $H_0: \fbeta=\bm{0}$. }
\label{powerMV}
\end{figure}
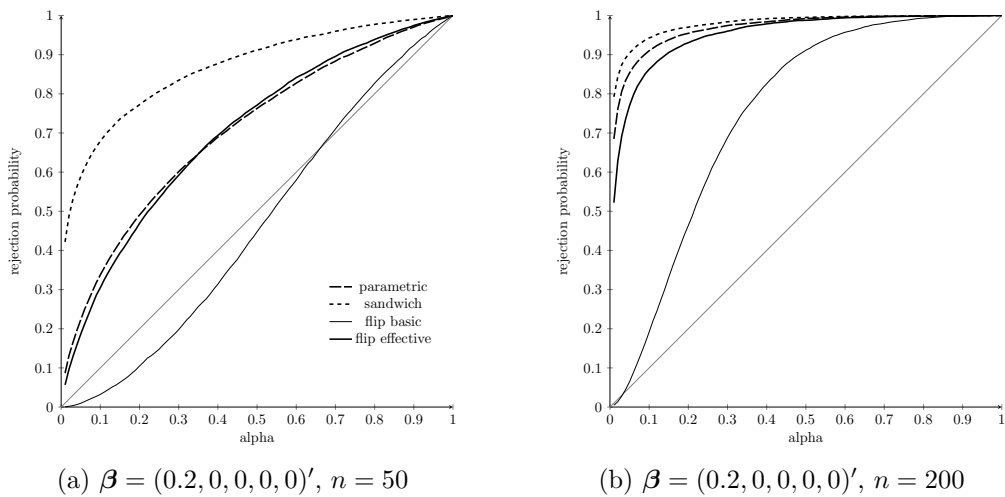

\section{Data analysis} \label{data}

We analyzed the dataset \emph{warpbreaks}. These data are  used in the example code of the \emph{gee} R package, available on CRAN.  
The dataset gives the number of warp breaks per loom, where a loom corresponds to a fixed length of yarn. There are 54 observations of 3 variables: the number of breaks, the type of wool (A or B) and the tension (low, medium or high). For each of the 6 possible combinations of wool and tension, there are 9 observations. Using various methods, we tested whether the number of breaks depends on the type of wool.

We first considered a basic Poisson model with
$$\log(\mu_i)= \gamma_1 + \beta \mathbbm{1}_{\{\text{wool}=B\}} + \gamma_2   \mathbbm{1}_{\{\text{tension}=\text{M}\}} + \gamma_3   \mathbbm{1}_{\{\text{tension}=\text{H}\}}.$$
The $\gamma_i$, $1\leq i \leq 3$, were nuisance parameters that were estimated using maximum likelihood.
We first tested $H_0:\beta=0$ using the parametric score test, obtaining a \emph{p}-value of $6.29 \cdot 10^{-5}$. (All tests performed were two-sided.) 

However, the data were clearly overdispersed: for each combination of wool and tension, the empirical variance of the 9 observations was substantially larger than the empirical mean. Thus the \emph{p}-value based on the parametric test had limited meaning. Fitting a quasi-Poisson model, which assumes constant overdispersion, gave a \emph{p}-value of 0.059.

As in Section \ref{simulations}, we also applied a Wald test, where we used a sandwich estimate \citep[][p. 280]{agresti2015foundations} of the variance of $\hat{\beta}$, to correct for the misspecified variance function. This resulted in a \emph{p}-value of $0.048.$

Further, we used the  sign-flipping test based on the basic scores $\nu_{\hat{\gamma},i}$, $i=1,...,54$ (still using the  basic Poisson model). We took $w=10^6$.
This resulted in a \emph{p}-value of $0.113$. This test is rather robust to model misspecification, but we know that it tends to be  conservative  when the score is correlated with the nuisance scores, as was the case here.

Finally, we performed the test of Theorem \ref{signflipeff} based on the effective score. 
This test is asymptotically exact under the correct model and has been shown to be robust against several forms of variance misspecification.
It provided a \emph{p}-value of $0.065$.

Based on this evidence, when maintaining a confidence level of $0.05$, it seems that we cannot reject $H_0$. Indeed, only the sandwich-based test provided a  \emph{p}-value below $0.05$, but this test is often anti-conservative, as discussed in Section \ref{overdEstnuis}.

\section*{Discussion}

We have proposed a test which relies on the assumption that individual score distributions are independent and have mean $0$ (in case of a point hypothesis) under the null. 
If the score contributions are misspecified due to overdispersion, heteroscedasticity or ignored nuisance covariates, then the traditional parametric tests lose their properties. 
The sign-flipping test is often robust to these types of misspecification and can still be asymptotically exact.

When nuisance parameters are estimated, the basic score contributions  become dependent. If a nuisance score is correlated with the score of the parameter of interest,
the estimation reduces the variance of the score, so that the sign-flipping test becomes conservative.
As a solution we propose to use the effective score, which is asymptotically the part of the score that is orthogonal to the nuisance score. 
The effective score is asymptotically unaffected by the  nuisance estimation, so that we again obtain an asymptotically exact test. 
We have proven that this is still the case when the scores and the Fisher information are misspecified by a constant, and simulations illustrate additional robustness.

When the parameter of interest is multi-dimensional, our test statistic involves a freely chosen matrix, which influences the power properties.
If this matrix is taken to be the inverse of the effective Fisher information and the assumed model is correct, then our test is asymptotically equivalent to the parametric score test.
Under the correct model, in certain situations our test is asymptotically equivalent to the global test \citep{goeman2006testing}, which is popular for testing hypotheses about high-dimensional parameters.

\appendix

\section{A lemma}

\begin{lemma} \label{mc}
Suppose that for $n\rightarrow \infty$, a vector $\bm{T}^n=(T_1^n,...,T_w^n)$ converges in distribution to a vector $\bm{T}$ of i.i.d. continuous variables. Then $\mathbb{P}(T_1^n>T_{[1-\alpha]}^n)\rightarrow \lfloor \alpha w\rfloor/w.$

\end{lemma}
\begin{proof}
Note that  $\mathbb{P}(T_1^n>T_{[1-\alpha]}^n)=\mathbb{P}(\bm{T}^n\in A)$, where 
$$A=\{(t_1,...,t_w)\in \mathbb{R}^w:  |  \{2\leq j \leq w: t_j<t_1\}   |  \geq \lceil (1-\alpha) w\rceil\}.$$

Let $\partial A$ be the boundary of $A$, i.e., the set of discontinuity points of $\mathbbm{1}_A$. 
Note that if $t\in \partial A$, then $t_i=t_j$ for some $1\leq i <j \leq w$.
It follows that $\mathbb{P}(\bm{T}\in \partial A)=0$. Since $\mathbbm{1}_A$ is continuous on $(\partial A)^c$, it follows from the continuous mapping theorem \citep[][Theorem 2.3]{van1998asymptotic} that $\mathbbm{1}_A(\bm{T}^n) \,{\buildrel d \over \rightarrow}\, \mathbbm{1}_A(\bm{T})$.

The elements of $\bm{T}$ are i.i.d. draws from the same distribution. Hence it follows from the Monte Carlo testing principle \citep{lehmann2005testing}  that under $H_0$, $\mathbb{P}(\bm{T}\in A)= \lfloor \alpha w \rfloor/w$. Thus $\mathbb{P}(\bm{T}^n\in A) \rightarrow \lfloor \alpha w \rfloor/w.$
\end{proof}

\section{Proofs of the results}

\noindent \textbf{Proof of Theorem \ref{mainST}}.
Suppose $H_0$ holds.
We will show that $\bm{T}^n=(T_1^n,...,T_w^n)$ converges in distribution to a multivariate normal distribution with mean $\bm{0}$ and variance $\lim_{n\rightarrow \infty} s^2_n \bm{I}$, where $\bm{I}$ is the $w\times w$ identity matrix.
It  then follows from Lemma \ref{mc} that $\mathbb{P}(T_1^n>T^n_{[1-\alpha]})\rightarrow \lfloor \alpha w \rfloor/w$.

Under $H_0$, for each $1\leq j \leq w$, $\mathbb{E}(T_j^n) =0$. For every  $1\leq j \leq w$, 
$var(T_j^n)=n^{-1}\sum_{i=1}^{n}var(\nu_i)= s^2_n$. 
Let $\bm{Q}_n$ be the covariance matrix of  $\bm{T}^n$. $\bm{Q}_n$ has zeroes off the diagonal. Indeed, for $1\leq j < k \leq w$, 
$$cov(T^n_j,T^n_k )=cov(n^{-1/2}\sum_{i=1}^n g_{ji}\nu_i,n^{-1/2}\sum_{i=1}^n g_{ki}\nu_i )=0,$$
since  the $g_{ki}$, $2\leq k \leq w$, are independent with mean $0$.
Hence $\bm{Q}_n$ converges to $\lim_{n\rightarrow\infty}s_n \bm{I}$. 
Note that $\bm{T}^n$ is a sum of $n$ vectors.
By the multivariate Lindeberg-Feller central limit theorem \citep{van1998asymptotic} $\bm{T}^n$ converges in distribution to a multivariate normal distribution with mean vector $\mathbf{0}$ and covariance matrix $\lim_{n\rightarrow\infty}s^2_n \bm{I}$. 

We have shown that $\bm{T}^n$ converges in distribution to a vector $\bm{T}$, say, of i.i.d. normal random variables. It now follows from Lemma \ref{mc} that $\mathbb{P}(T_1^n>T_{[1-\alpha]}^n)\rightarrow \lfloor \alpha w\rfloor/w.$ \qed

\phantom{.}
\noindent \textbf{Proof of Proposition \ref{exacttest}}. Note that $(\nu_1,...,\nu_n) \,{\buildrel d \over =}\, (g_{j1}\nu_1,...,g_{jn}\nu_n)$ for every $1 \leq j \leq w$.
This means that the test becomes a basic random transformation test and the results follow  from  the proof of Theorem 2  in \citet{hemerik2018false}.

\phantom{.}
\noindent \textbf{Proof of Theorem  \ref{signflipeff}}.
Suppose that $H_0$ holds.
Note that 
$$ S^*_{\hat{\fgamma}} = S_{\hat{\fgamma}}  -\Ima_{12}'\Ima_{22}^{-1} \fatS^{(k-1)}_{\hat{\fgamma}}=  S_{\hat{\fgamma}}  -\Imag_{12}'\Imag_{22}^{-1} \fatS^{(k-1)}_{\hat{\fgamma}} +o_{\mathbb{P}_{\beta_0,\fgamma_0}}(1)=$$
$$ S_{\fgamma_0} - \Imag_{12}' \sqrt{n}(\hat{\fgamma}-\fgamma_0) 
- \Imag_{12}' \Imag_{22}^{-1}\Big \{ \fatS^{(k-1)}_{\fgamma_0} - \Imag_{22} \sqrt{n}(\hat{\fgamma}-\fgamma_0) \Big\} +o_{\mathbb{P}_{\beta_0,\fgamma_0}}(1)=
$$
$$S^*_{\fgamma_0}+ o_{\mathbb{P}_{\beta_0,\fgamma_0}}(1).$$

Let $2\leq j \leq w$ and
$$S^{j+}_{{\fgamma}} =n^{-1/2}\sum_{i=1}^n \bm{1}_{\{g_{ji=1}\}} \nu_{{\fgamma},i}, \quad S^{j-}_{{\fgamma}} =n^{-1/2}\sum_{i=1}^n \bm{1}_{\{g_{ji=-1}\}} \nu_{{\fgamma},i}.$$
Note that $$ S^{j}_{\hat{\fgamma}}  = S^{j+}_{\hat{\fgamma}} - S^{j-}_{\hat{\fgamma}}=
 \Big \{S^{j+}_{{\fgamma_0}} - \frac{1}{2}\sqrt{n}\Imag_{12}'(\hat{\fgamma}- {\fgamma_0})\Big \}    - \Big\{ S^{j-}_{{\fgamma_0}}   - \frac{1}{2}\sqrt{n}\Imag_{12}'     (\hat{\fgamma}- {\fgamma_0}) \Big \}    +o_{\mathbb{P}_{\beta_0,\fgamma_0}}(1)= $$ 
$$S^{j+}_{{\fgamma_0}} - S^{j-}_{\hat{\fgamma}}+o_{\mathbb{P}_{\beta_0,\fgamma_0}}(1)=  S^{j}_{{\fgamma_0}}+o_{\mathbb{P}_{\beta_0,\fgamma_0}}(1). $$
The intuitive reason why $S^{j}_{\hat{\fgamma}}=S^{j}_{{\fgamma_0}}+o_{\mathbb{P}_{\beta_0,\fgamma_0}}(1)$, is that  the estimation of $\hat{\fgamma}$ does not cause the summands underlying $S^{j}_{\hat{\fgamma}}$ to be correlated.
Similarly we find that $\fatS^{(k-1),j}_{\hat{\fgamma}}= \fatS^{(k-1),j}_{{\fgamma_0}}+o_{\mathbb{P}_{\beta_0,\fgamma_0}}(1)$
and conclude that
$S^{*j}_{\hat{\fgamma}}= S^{*j}_{{\fgamma_0}}+o_{\mathbb{P}_{\beta_0,\fgamma_0}}(1).$

Let $\bm{T}^{n}$ be as in the proof of Theorem \ref{mainST}, with $\nu_i$ replaced by 
$\nu^*_{{\fgamma_0},i} .$
Suppose $H_0$ holds and $\Ima=\Imag$, so that  the summands underlying $T_j^{n}$ are independent.
For every $1\leq i\leq n$,  $\mathbb{E}(\nu^*_{{\fgamma_0},i})=0.$ 
The elements of $\bm{T}^{n}$ are uncorrelated  and  have common  variance $ var(\nu^*_{\fgamma_0,1}).$
By the multivariate  central limit theorem \citep{van1998asymptotic, greene2012econometric}, $\bm{T}^{n}$ converges in distribution to $N(\bm{0},  var(\nu^*_{\fgamma_0,1}) \bm{I} )$. We supposed that $\Ima=\Imag$ to use the central limit theorem, but the asymptotic distribution of $\bm{T}^{n}$ is the same if $\Ima$ is any  consistent estimator of $\Imag$.

Let $\hat{\bm{T}}^{n}$ be as in the proof of Theorem \ref{mainST}, with $\nu_i$ replaced by $\nu^*_{\hat{\fgamma},i}.$
For every $1\leq j \leq w,$ $S^{*j}_{\hat{\fgamma}}= S^{*j}_{\fgamma_0}+o_{\mathbb{P}_{\beta_0,\fgamma_0}}(1)$. Thus $\hat{\bm{T}}^{n}$ and $\bm{T}^{n}$ are asymptotically equivalent. The result now follows from Lemma \ref{mc}.  \qed

\phantom{.}
\noindent \textbf{Proof of Proposition \ref{aseqeff}}.
 For $2\leq j\leq w$ consider
$$S^{*j+}_{\hat{\fgamma}} =n^{-1/2}\sum_{i=1}^n \bm{1}_{\{g_{ji=1}\}} \nu^*_{\hat{\fgamma},i},\quad  S^{*j-}_{\hat{\fgamma}} =n^{-1/2}\sum_{i=1}^n \bm{1}_{\{g_{ji=-1}\}} \nu^*_{\hat{\fgamma},i}$$
$$S^{(k-1),j+}_{\hat{\fgamma}}=n^{-1/2}\sum_{i=1}^n \bm{1}_{\{g_{ji=1}\}} \nu^{(k-1)}_{\hat{\fgamma},i}, \quad S^{(k-1),j-}_{\hat{\fgamma}}=n^{-1/2}\sum_{i=1}^n \bm{1}_{\{g_{ji=-1}\}} \nu^{(k-1)}_{\hat{\fgamma},i}.$$
We have 
$$S^{*j+}_{\hat{\fgamma}} = S^{j+}_{\hat{\fgamma}} - \Ima_{12}' \Ima_{22}^{-1} S^{(k-1),j+}_{\hat{\fgamma}}= $$ 
$$S^{j+}_{{\fgamma_0}}-\frac{1}{2}\sqrt{n} \Imag_{12}' (\hat{\fgamma}- {\fgamma_0})-  \Imag_{12}' \Imag_{22}^{-1}\big\{  S^{(k-1),j+}_{{\fgamma_0}}-\frac{1}{2}\sqrt{n} \Imag_{22} (\hat{\fgamma}- {\fgamma_0})    \big\}+o_{\mathbb{P}_{\beta^n,\fgamma_0}}(1)= $$ 
$$S^{*j+}_{{\fgamma_0}}+o_{\mathbb{P}_{\beta^n,\fgamma_0}}(1)$$ 
 and analogously  $S^{*j-}_{\hat{\fgamma}} =S^{*j-}_{{\fgamma_0}}+o_{\mathbb{P}_{\beta^n,\fgamma_0}}(1)$.
By \citet[][p. 341]{marohn2002comment}, for $2\leq j \leq w$, $S^{*j+}_{{\fgamma_0}}$ and  $S^{*j-}_{{\fgamma_0}}$ have an asymptotic $N(\frac{1}{2}\xi\sigma_0^2,\frac{1}{2}\sigma_0^2)$ distribution. 
 Since they are independent, it follows that  $T_j^n=S^{*j+}_{\hat{\fgamma}}-S^{*j-}_{\hat{\fgamma}}$ has an asymptotic $N(0,\sigma_0^2)$ distribution, $2\leq j \leq w$. With  the multivariate central limit theorem we find that  $(T_2^n,...,T_w^n)$ converges in distribution to a vector of $w-1$  i.i.d.  $N(0,\sigma_0^2)$ variables as $n\rightarrow\infty$.

Let $\epsilon$, $\epsilon'>0$. 
Let $(T_1',...,T_w')$ have the asymptotic distribution of $(T_1^n,...,T_w^n)$.
Let $(T_1'',...,T_w'')$ be a vector of  $w$ i.i.d.  $N(0,\sigma_0^2)$ variables. Apart from the first element, these two vectors have the same distribution.
For $w\in\{2,3,...\}$, define $T^{[w]}_{[1-\alpha]}$ like $T^n_{[1-\alpha]}$, but based on the values  $T_1',...,T_{w}'$ instead of $T_1^n,...,T_w^n$.
Also define $T^{[[w]]}_{[1-\alpha]}$ like $T^n_{[1-\alpha]}$, but based on  the values $T_1'',...,T_{w}''$.
Note that as $w\rightarrow\infty$, the empirical quantile $T^{[[w]]}_{[1-\alpha]}$ converges in distribution to the constant $\sigma_0\Phi(1-\alpha)$. Further note that for $w\rightarrow\infty$, $T^{[[w]]}_{[1-\alpha]}-T^{[w]}_{[1-\alpha]}$ converges in distribution to $0$.
Thus there is a $W  \in\mathbb{N}$ such that for all $w>W$, 
\begin{equation}\label{P1}
\mathbb{P}(|T^{[w]}_{[1-\alpha]}-\sigma_0\Phi(1-\alpha)|<\epsilon')>1-\epsilon.
\end{equation}

Since the distribution of $(T_1^n,...,T_w^n)$ converges to the distribution of $(T_1',...,T_w')$  as $n\rightarrow\infty$,  
\begin{equation}\label{convdistr}
T^n_{[1-\alpha]}    \xrightarrow[]{d}   T^{[w]}_{[1-\alpha]}
\end{equation}
 as $n\rightarrow\infty$.  
Since in the present proof $w$ is not fixed, we will write $T^n_{[1-\alpha]}=T^{n,w}_{[1-\alpha]}$. 
By results \eqref{P1} and \eqref{convdistr}, for  $w>W$, $\liminf_{n\rightarrow\infty}\mathbb{P}(|T^{n,w}_{[1-\alpha]}-\sigma_0\Phi(1-\alpha)|<\epsilon')>1-\epsilon$.
Thus $\lim_{w\rightarrow\infty}\liminf_{n\rightarrow\infty}\mathbb{P}(|T^{n,w}_{[1-\alpha]}-\sigma_0\Phi(1-\alpha)|<\epsilon')=1$.

The distribution of $T_1^{n}$, which does not depend on $w$, 
converges to a continuous distribution as $n\rightarrow\infty$.
It follows that for every $\epsilon'' >0$, there is an $W'$ such that there is a $N$ such that for all $w>W'$ and $n>N$,
$\mathbb{E}\big( | \mathbbm{1}_{\{ T_1^n>T^{n,w}_{[1-\alpha]} \}} -\mathbbm{1}_{\{ T_1^n>\sigma_0\Phi(1-\alpha) \}}  |   \big)<\epsilon''$.
This means that $\lim_{w\rightarrow\infty}\liminf_{n\rightarrow\infty}\mathbb{E}\big( | \mathbbm{1}_{\{ T_1^n>T^{n,w}_{[1-\alpha]} \}} -\mathbbm{1}_{\{ T_1^n>\sigma_0\Phi(1-\alpha) \}}  |   \big)=0$, as was to be shown.    \qed

\phantom{.}
\noindent \textbf{Proof of Proposition \ref{missp}}.
For every $1\leq j \leq w$ we have 
$$\tilde{S}^{*j}_{\hat{\fgamma}}=  c_1 S^j_{\hat{\fgamma}}-c_2\Ima_{12}' c_2^{-1} \Ima_{22}^{-1} c_1\mathbf{S}^{(k-1),j}_{\hat{\fgamma}}= c_1 S^{*j}_{\hat{\fgamma}}.$$ Hence the test is identical to that of Theorem \ref{signflipeff}, since that test is unchanged if all $T_j^n$, $1\leq j\leq w$, are multiplied by the same constant.   \qed

\phantom{.}
\noindent \textbf{Proof of Theorem \ref{mvknownn}}.
Suppose $H_0$ holds. Consider the $d\times j$-matrix
\begin{equation} \label{Sjmv}
\Big ( n^{1/2} \sum_{i=1}^n g_{ji} \fnu_{\gamma_0,i}\Big)_{1\leq j \leq w}.
\end{equation}
It follows from the multivariate central limit theorem \citep{van1998asymptotic} that, as $n\rightarrow \infty$,  this matrix converges in distribution to a matrix with 
identically distributed columns which are independent of each other.
Note that for every $1\leq j \leq w$, $T_j^n$ is a function of the $j$-th column of the matrix \eqref{Sjmv}.
Thus, with the continuous mapping theorem \citep[][Theorem 2.3]{van1998asymptotic} it follows that $(T_1^n,...,T_j^n)$ also  converges in distribution to a vector with continuous i.i.d. elements.  The result now follows from Lemma \ref{mc}.   \qed

\phantom{.}
\noindent \textbf{Proof of Theorem \ref{mvunknownn}}.
Consider the case $\hat{\fgamma}=\fgamma_0$. 
As in the proof of Theorem \ref{mvknownn}, under $H_0$, $(T_1^n,...,T_w^n)$ converges in distribution to a vector  of $w$ i.i.d. variables.
As in the proof of Theorem  \ref{signflipeff}, the same is true if we take $\hat{\fgamma}$ to be a different $\sqrt{n}$-consistent estimator of $\fgamma_0$. (Again, the reason is that the effective score based on $\hat{\fgamma}$ is asymptotically equivalent to the effective score based on $\gamma_0$.)
The result now follows from Lemma \ref{mc} again.  \qed

\phantom{.}
\noindent \textbf{Proof of Proposition \ref{aseqmv}}.
By \citet{hall1990large},  $n^{-1/2}\sum_{i=1}^n \fnu_{\hat{\fgamma},i}^{*}$ has an asymptotic $N(\bm{0},\Imag^*)$ distribution under $\bm{\beta}=\bm{\beta}_0$.
 Analogously to  the one-dimensional case at Proposition \ref{aseqeff}, for $2\leq j \leq w$, the vector $n^{-1/2}\sum_{i=1}^n g_{ji}\fnu_{\hat{\fgamma},i}^{*}$  is asymptotically the difference of two mutually  independent  $N(\frac{1}{2}\Imag^* \bm{\xi} , \frac{1}{2}\Imag^*)$  vectors \citep{hall1990large},  so that it also  
 has an  asymptotic $N(\bm{0},\Imag^*)$ distribution (under $\fbeta=\fbeta^n$). 
As in the proof of Theorem \ref{mvknownn}, by the multivariate central limit theorem, the $d \times (w-1)$ matrix 
 $\big (n^{-1/2}\sum_{i=1}^n g_{ji}\fnu_{\hat{\fgamma},i}^{*} \big )_{2\leq j \leq w}$ converges to a   matrix with $w-1$ independent $N(\bm{0},\Imag^*)$ columns as $n\rightarrow\infty $. Hence, by the continuous mapping theorem, as $n\rightarrow\infty$, $(T_2^n,...,T_w^n)$ converges in distribution to a vector of $w-1$ i.i.d. variables (under $\fbeta=\fbeta^n$), which follow the asymptotic distribution which $T_1^n$ has under $\bm{\beta}=\bm{\beta}_0$.

The result now follows as at the end of the proof of Proposition \ref{aseqeff}. \qed

\setlength{\bibsep}{3pt plus 0.3ex}  
\def\bibfont{\small}  

\bibliographystyle{biblstyle}
\bibliography{references}

\begin{thebibliography}{36}
\providecommand{\natexlab}[1]{#1}
\providecommand{\url}[1]{\texttt{#1}}
\expandafter\ifx\csname urlstyle\endcsname\relax
  \providecommand{\doi}[1]{doi: #1}\else
  \providecommand{\doi}{doi: \begingroup \urlstyle{rm}\Url}\fi

\bibitem[Agresti(2015)]{agresti2015foundations}
Agresti, A.
\newblock \emph{Foundations of linear and generalized linear models}.
\newblock John Wiley \& Sons, 2015.

\bibitem[Boos(1992)]{boos1992on}
Boos, D.~D.
\newblock On generalized score tests.
\newblock \emph{The American Statistician}, 46:\penalty0 327--333, 1992.

\bibitem[Canay et~al.(2017)Canay, Romano, and Shaikh]{canay2017randomization}
Canay, I.~A., Romano, J.~P., and Shaikh, A.~M.
\newblock Randomization tests under an approximate symmetry assumption.
\newblock \emph{Econometrica}, 85\penalty0 (3):\penalty0 1013--1030, 2017.

\bibitem[Chung and Romano(2013)]{chung2013exact}
Chung, E.~Y. and Romano, J.~P.
\newblock Exact and asymptotically robust permutation tests.
\newblock \emph{The Annals of Statistics}, 41\penalty0 (2):\penalty0 484--507,
  2013.

\bibitem[Cox and Hinkley(1979)]{cox1979theoretical}
Cox, D.~R. and Hinkley, D.~V.
\newblock \emph{Theoretical statistics}.
\newblock CRC Press, 1979.

\bibitem[Cox and Reid(1987)]{cox1987parameter}
Cox, D.~R. and Reid, N.
\newblock Parameter orthogonality and approximate conditional inference.
\newblock \emph{Journal of the Royal Statistical Society. Series B
  (Methodological)}, pages 1--39, 1987.

\bibitem[Fisher(1935)]{fisher1935}
Fisher, R.~A.
\newblock \emph{The design of experiments}.
\newblock Oliver and Boyd, 1935.

\bibitem[Freedman(2006)]{freedman2006so}
Freedman, D.~A.
\newblock On the so-called ``{H}uber sandwich estimator'' and ``robust standard
  errors''.
\newblock \emph{The American Statistician}, 60\penalty0 (4):\penalty0 299--302,
  2006.

\bibitem[Ganong and J{\"a}ger(2018)]{ganong2018permutation}
Ganong, P. and J{\"a}ger, S.
\newblock A permutation test for the regression kink design.
\newblock \emph{Journal of the American Statistical Association}, 113\penalty0
  (522):\penalty0 494--504, 2018.

\bibitem[Goeman et~al.(2004)Goeman, Van De~Geer, De~Kort, and
  Van~Houwelingen]{goeman2004global}
Goeman, J.~J., Van De~Geer, S.~A., De~Kort, F., and Van~Houwelingen, H.~C.
\newblock A global test for groups of genes: testing association with a
  clinical outcome.
\newblock \emph{Bioinformatics}, 20\penalty0 (1):\penalty0 93--99, 2004.

\bibitem[Goeman et~al.(2006)Goeman, Van De~Geer, and
  Van~Houwelingen]{goeman2006testing}
Goeman, J.~J., Van De~Geer, S.~A., and Van~Houwelingen, H.~C.
\newblock Testing against a high dimensional alternative.
\newblock \emph{Journal of the Royal Statistical Society: Series B (Statistical
  Methodology)}, 68\penalty0 (3):\penalty0 477--493, 2006.

\bibitem[Goeman et~al.(2011)Goeman, Van~Houwelingen, and
  Finos]{goeman2011testing}
Goeman, J.~J., Van~Houwelingen, H.~C., and Finos, L.
\newblock Testing against a high-dimensional alternative in the generalized
  linear model: asymptotic type {I} error control.
\newblock \emph{Biometrika}, 98\penalty0 (2):\penalty0 381--390, 2011.

\bibitem[Greene(2012)]{greene2012econometric}
Greene, W.~H.
\newblock \emph{Econometric analysis}.
\newblock Harlow: Pearson Education Limited, 2012.

\bibitem[Hall and Mathiason(1990)]{hall1990large}
Hall, W. and Mathiason, D.~J.
\newblock On large-sample estimation and testing in parametric models.
\newblock \emph{International Statistical Review/Revue Internationale de
  Statistique}, 58\penalty0 (1):\penalty0 77--97, 1990.

\bibitem[Hemerik et~al.(2019)Hemerik, Solari, and
  Goeman]{hemerik2019permutation}
Hemerik, J., Solari, A., and Goeman, J.
\newblock Permutation-based simultaneous confidence bounds for the false
  discovery proportion.
\newblock \emph{Biometrika}, 106\penalty0 (3):\penalty0 635--649, 2019.

\bibitem[Hemerik and Goeman(2018{\natexlab{a}})]{hemerik2018exact}
Hemerik, J. and Goeman, J.~J.
\newblock Exact testing with random permutations.
\newblock \emph{TEST}, 27\penalty0 (4):\penalty0 811--825, 2018{\natexlab{a}}.

\bibitem[Hemerik and Goeman(2018{\natexlab{b}})]{hemerik2018false}
Hemerik, J. and Goeman, J.~J.
\newblock False discovery proportion estimation by permutations: confidence for
  significance analysis of microarrays.
\newblock \emph{Journal of the Royal Statistical Society: Series B (Statistical
  Methodology)}, 80\penalty0 (1):\penalty0 137--155, 2018{\natexlab{b}}.

\bibitem[Kauermann and Carroll(2000)]{kauermann2000sandwich}
Kauermann, G. and Carroll, R.~J.
\newblock The sandwich variance estimator: Efficiency properties and coverage
  probability of confidence intervals.
\newblock 2000.

\bibitem[Lehmann and Romano(2005)]{lehmann2005testing}
Lehmann, E.~L. and Romano, J.~P.
\newblock \emph{Testing statistical hypotheses}.
\newblock Springer Science \& Business Media, 2005.

\bibitem[Maas and Hox(2004)]{maas2004robustness}
Maas, C.~J. and Hox, J.~J.
\newblock Robustness issues in multilevel regression analysis.
\newblock \emph{Statistica Neerlandica}, 58\penalty0 (2):\penalty0 127--137,
  2004.

\bibitem[Marohn(2002)]{marohn2002comment}
Marohn, F.
\newblock A comment on locally most powerful tests in the presence of nuisance
  parameters.
\newblock \emph{Communications in Statistics-Theory and Methods}, 31\penalty0
  (3):\penalty0 337--349, 2002.

\bibitem[Marriott(1979)]{marriott1979barnard}
Marriott, F. H.~C.
\newblock Barnard's {M}onte {C}arlo tests: How many simulations?
\newblock \emph{Journal of the Royal Statistical Society. Series C (Applied
  Statistics)}, 28\penalty0 (1):\penalty0 75--77, 1979.

\bibitem[Pauly et~al.(2015)Pauly, Brunner, and
  Konietschke]{pauly2015asymptotic}
Pauly, M., Brunner, E., and Konietschke, F.
\newblock Asymptotic permutation tests in general factorial designs.
\newblock \emph{Journal of the Royal Statistical Society: Series B (Statistical
  Methodology)}, 77\penalty0 (2):\penalty0 461--473, 2015.

\bibitem[Pesarin(2015)]{pesarin2015some}
Pesarin, F.
\newblock Some elementary theory of permutation tests.
\newblock \emph{Communications in Statistics-Theory and Methods}, 44\penalty0
  (22):\penalty0 4880--4892, 2015.

\bibitem[Pesarin(2001)]{pesarin2001multivariate}
Pesarin, F.
\newblock \emph{Multivariate permutation tests: with applications in
  biostatistics}, volume 240.
\newblock Wiley Chichester, 2001.

\bibitem[Pesarin and Salmaso(2010{\natexlab{a}})]{pesarin2010finite}
Pesarin, F. and Salmaso, L.
\newblock Finite-sample consistency of combination-based permutation tests with
  application to repeated measures designs.
\newblock \emph{Journal of Nonparametric Statistics}, 22\penalty0 (5):\penalty0
  669--684, 2010{\natexlab{a}}.

\bibitem[Pesarin and Salmaso(2010{\natexlab{b}})]{pesarin2010permutation}
Pesarin, F. and Salmaso, L.
\newblock \emph{Permutation tests for complex data: theory, applications and
  software}.
\newblock John Wiley \& Sons, 2010{\natexlab{b}}.

\bibitem[Rao(1948)]{rao1948large}
Rao, C.~R.
\newblock Large sample tests of statistical hypotheses concerning several
  parameters with applications to problems of estimation.
\newblock In \emph{Mathematical Proceedings of the Cambridge Philosophical
  Society}, volume~44, pages 50--57. Cambridge Univ Press, 1948.

\bibitem[Rayner(1997)]{rayner1997asymptotically}
Rayner, J.
\newblock The asymptotically optimal tests.
\newblock \emph{Journal of the Royal Statistical Society: Series D (The
  Statistician)}, 46\penalty0 (3):\penalty0 337--345, 1997.

\bibitem[Rippon and Rayner(2010)]{rippon2010generalised}
Rippon, P. and Rayner, J.~C.
\newblock Generalised score and {W}ald tests.
\newblock \emph{Advances in Decision Sciences}, 2010.

\bibitem[Solari et~al.(2014)Solari, Finos, and Goeman]{solari2014rotation}
Solari, A., Finos, L., and Goeman, J.~J.
\newblock Rotation-based multiple testing in the multivariate linear model.
\newblock \emph{Biometrics}, 70\penalty0 (4):\penalty0 954--961, 2014.

\bibitem[Tusher et~al.(2001)Tusher, Tibshirani, and
  Chu]{tusher2001significance}
Tusher, V.~G., Tibshirani, R., and Chu, G.
\newblock Significance analysis of microarrays applied to the ionizing
  radiation response.
\newblock \emph{Proceedings of the National Academy of Sciences}, 98\penalty0
  (9):\penalty0 5116--5121, 2001.

\bibitem[Van~der Vaart(1998)]{van1998asymptotic}
Van~der~ Vaart, A.~W.
\newblock \emph{Asymptotic statistics}, volume~3.
\newblock Cambridge university press, 1998.

\bibitem[Westfall and Young(1993)]{westfall1993resampling}
Westfall, P.~H. and Young, S.~S.
\newblock \emph{Resampling-based multiple testing: Examples and methods for
  p-value adjustment}.
\newblock John Wiley \& Sons, 1993.

\bibitem[Winkler et~al.(2014)Winkler, Ridgway, Webster, Smith, and
  Nichols]{winkler2014permutation}
Winkler, A.~M., Ridgway, G.~R., Webster, M.~A., Smith, S.~M., and Nichols,
  T.~E.
\newblock Permutation inference for the general linear model.
\newblock \emph{Neuroimage}, 92:\penalty0 381--397, 2014.

\bibitem[Winkler et~al.(2016)Winkler, Ridgway, Douaud, Nichols, and
  Smith]{winkler2016faster}
Winkler, A.~M., Ridgway, G.~R., Douaud, G., Nichols, T.~E., and Smith, S.~M.
\newblock Faster permutation inference in brain imaging.
\newblock \emph{NeuroImage}, 141:\penalty0 502--516, 2016.

\end{thebibliography}

\end{document}